\documentclass[11pt]{article}
\usepackage[margin=1in]{geometry}
\usepackage[T1]{fontenc}   
\usepackage{fullpage}
\usepackage{soul}

\usepackage{hyperref}
\usepackage{mathrsfs, amsmath, amssymb, amsthm, amsfonts, graphicx, color, subcaption, enumerate, bm, cleveref, enumitem, array, mathtools, tikz}
\usetikzlibrary{arrows.meta, positioning}
\usepackage[ruled,linesnumbered]{algorithm2e}

\usepackage{framed}
\usepackage[framemethod=tikz]{mdframed}

\SetKwIF{OnEvent}{OnElseIf}{OnElse}{on event}{do}{on event else if}{else}{end}
\SetKwIF{WaitUntil}{WaitElseIf}{WaitElse}{wait until}{then}{else if}{else}{end}
\SetKwInOut{Requires}{Requires}
\crefname{algocf}{algorithm}{algorithms}
\Crefname{algocf}{Algorithm}{Algorithms}

\usepackage[style=trad-alpha,natbib=true,maxcitenames=4]{biblatex}
\usepackage{thm-restate}

\usepackage[colorinlistoftodos,prependcaption,textsize=tiny]{todonotes}

\newcommand{\leader}{\mathsf{Leader}}
\newcommand{\nonleader}{\mathsf{Non\text{-}Leader}}
\newcommand{\isleader}{\mathsf{isLeader}}
\newcommand{\counter}{\mathsf{Count}}

\newcommand{\id}{\mathsf{ID}}

\newcommand{\idmax}{\mathsf{ID_{\max}}}
\newcommand{\idmin}{\mathsf{ID_{\min}}}

\newcommand{\diff}{\mathsf{Diff}}

\newcommand{\length}{\mathsf{length}}

\newcommand{\notify}{\mathsf{Notify}}
\newcommand{\parentport}{\mathsf{Parent}}
\newcommand{\state}{\mathsf{State}}
\newcommand{\start}{\mathsf{StartDFS}}
\newcommand{\sendall}{\mathsf{SendAll}}
\newcommand{\sendpulsesuntil}{\mathsf{SendPulsesUntil}}
\newcommand{\sendpulses}{\mathsf{SendPulses}}
\newcommand{\rcvpulse}{\mathsf{RcvPulse}}
\newcommand{\sendexplore}{\mathsf{SendExplore}}
\newcommand{\senddone}{\mathsf{SendDone}}
\newcommand{\receiveexplore}{\mathsf{ReceiveExplore}}
\newcommand{\receivedone}{\mathsf{ReceiveDone}}

\newcommand{\send}{\mathsf{Send}}
\newcommand{\receive}{\mathsf{Receive}}

\newcommand{\msg}{\texttt{p}\xspace}

\newcommand{\IH}{\mathsf{IH}}

\newtheorem{theorem}{Theorem}

\newtheorem{lemma}[theorem]{Lemma}

\newtheorem{observation}[theorem]{Observation}

\newtheorem{corollary}[theorem]{Corollary}

\newtheorem{proposition}[theorem]{Proposition}

\newtheorem{question}{Question}

\newtheorem{definition}[theorem]{Definition}
\title{Content-Oblivious Leader Election \\ in 2-Edge-Connected Networks}

\author{\hspace{1cm}
Jérémie {Chalopin}\footnote{Aix-Marseille University, CNRS, LIS, Marseille, France. ORCID: 0000-0002-2988-8969. Email: jeremie.chalopin@lis-lab.fr} \and
 Yi-Jun Chang\footnote{National University of Singapore, Singapore. ORCID: 0000-0002-0109-2432. Email: cyijun@nus.edu.sg} 
 \and Lyuting Chen\footnote{National University of Singapore, Singapore. ORCID: 0009-0002-8836-6607. Email: e0726582@u.nus.edu}\hspace{1cm}
 \and Giuseppe A. {Di Luna}\footnote{DIAG, Sapienza University of Rome, Italy. 
 Email: diluna@diag.uniroma1.it}
 \and Haoran Zhou\footnote{National University of Singapore, Singapore. ORCID: 0009-0001-2458-5344. Email: haoranz@u.nus.edu}
}

\date{}

\begin{document}

\maketitle

\begin{abstract}
Censor-Hillel, Cohen, Gelles, and Sela (PODC 2022 \& Distributed Computing 2023) studied \emph{fully-defective} asynchronous networks, where communication channels may suffer an extreme form of alteration errors, rendering messages completely corrupted. The model is equivalent to \emph{content-oblivious} computation, where nodes communicate solely via pulses. They showed that if the network is 2-edge-connected, then any algorithm for a noiseless setting can be simulated in the fully-defective setting; otherwise, no non-trivial computation is possible in the fully-defective setting. However, their simulation requires a predesignated leader, which they conjectured to be necessary for any non-trivial content-oblivious task.

Recently, Frei, Gelles, Ghazy, and Nolin (DISC 2024) refuted this conjecture for the special case of oriented ring topology. They designed two asynchronous  content-oblivious leader election algorithms with message complexity $O(n \cdot \idmax)$, where $n$ is the number of nodes and $\idmax$ is the maximum $\id$. The first algorithm stabilizes in unoriented rings without termination detection. The second algorithm quiescently terminates in oriented rings, thus enabling the execution of the simulation algorithm after leader election. 

In this work, we present two results:
\begin{description}
    \item[General 2-edge-connected topologies:] First, we show an asynchronous content-oblivious leader election algorithm that quiescently terminates in any 2-edge-connected network with message complexity $O(m \cdot N \cdot \idmin)$, where $m$ is the number of edges,  $N$ is a known upper bound on the number of nodes, and $\idmin$ is the smallest $\id$.
    
    Combined with the above simulation, this result shows that whenever a size bound $N$ is known, any noiseless algorithm can be simulated in the fully-defective model without a preselected leader, fully refuting the conjecture.

    \item[Unoriented rings:] We then show that the knowledge of $N$ can be dropped in unoriented ring topologies by presenting a quiescently terminating election algorithm with message complexity $O(n \cdot \id_{\max})$ that matches the previous bound.  
    
    Consequently, this result constitutes a \emph{strict} improvement over the previous state of the art and shows that, on rings, fully-defective and noiseless communication are computationally equivalent, with no additional assumptions.
\end{description}

\end{abstract}
\thispagestyle{empty}
\newpage
\tableofcontents
\thispagestyle{empty}
\newpage
\pagenumbering{arabic}

\section{Introduction}

Fault tolerance is a cornerstone of distributed computing, enabling systems to remain operational despite various failures, such as node crashes or channel noise~\cite{Barborak1993,Dubrova2013,Koren2020,Raynal2018}. 
One important category of faults, known as \emph{alteration errors}, changes the content of messages but does not create new messages or eliminate existing ones. Traditional methods for handling such errors often rely on adding redundancy through coding techniques and assume known bounds on the frequency or severity of faults. However, in many practical settings, such assumptions may not hold, limiting the effectiveness of these approaches.

\paragraph{Fully-defective networks.}  Censor-Hillel, Cohen, Gelles, and Sela~\cite{fully-defective-22} introduced the model of \emph{fully-defective} asynchronous networks, in which all links are subject to a severe form of alteration errors: The content of every message can be arbitrarily modified, although the adversary cannot insert new messages or remove existing messages. Since messages can no longer carry meaningful information, the model is equivalent to \emph{content-oblivious} computation, where communication is reduced to sending and receiving \emph{pulses}. Algorithms in this setting operate entirely based on the patterns of pulse arrivals.

\paragraph{An impossibility result.} A natural approach to designing a content-oblivious algorithm is to use unary encoding, representing messages as sequences of pulses. However, in the \emph{asynchronous} setting, where no upper bound on message delivery time is known, there is no reliable way to detect the end of a sequence. Censor-Hillel, Cohen, Gelles, and Sela~\cite{fully-defective-22} established a strong \emph{impossibility} result: Let \( f(x, y) \) be any non-constant function, then any two-party deterministic asynchronous content-oblivious communication protocol, where one party holds \( x \) and the other holds \( y \), must fail to compute \( f \) correctly.

 The impossibility result implies that, in a certain sense, no non-trivial computation is possible in any network that is \emph{not} 2-edge-connected---that is, any network containing an edge $e$ whose removal disconnects the graph into two components. By assigning control of each component to a different party, the problem reduces to the two-party setting to which the impossibility result applies.

\paragraph{Algorithm simulation over 2-edge-connected networks.} 
Interestingly, Censor-Hillel, Cohen, Gelles, and Sela~\cite{fully-defective-22} showed that in 2-edge-connected networks, any algorithm designed for the noiseless setting can be simulated in a fully defective network. We briefly sketch the main ideas underlying their approach below.

For the special case of an oriented ring, the challenge of termination detection can be addressed. Suppose a node wishes to broadcast a message. It can send a sequence of pulses that represents the unary encoding of the message in one direction and use the opposite direction to signal termination. To allow all nodes the opportunity to communicate, a token can be circulated in the ring, enabling each node to speak in turn.

To extend this approach to general 2-edge-connected networks, they leverage a result by Robbins~\cite{robbins1939theorem}, which states that any 2-edge-connected undirected graph can be oriented to form a strongly connected directed graph. This guarantees the existence of a \emph{Robbins cycle}, which is an oriented ring that traverses all nodes, possibly revisiting some nodes multiple times. They designed a content-oblivious algorithm that uses a depth-first search (DFS) to construct such a cycle.

\paragraph{Preselected leader assumption.} 
A key limitation of their result is that it inherently relies on a \emph{preselected leader}. For instance, a leader is needed to generate a unique token in a ring or to select the root for the DFS. Censor-Hillel, Cohen, Gelles, and Sela~\cite{fully-defective-22} conjectured that the preselected leader assumption is essential for any non-trivial content-oblivious computation, leaving the following question open.

\begin{question}\label{Q1}
In 2-edge-connected networks, is it possible to simulate any algorithm designed for a noiseless setting in the fully defective model without relying on a preselected leader?
\end{question}

\paragraph{Content-oblivious leader election on rings.} 
To answer the above question affirmatively, it suffices to design a content-oblivious leader election algorithm that possesses certain desirable properties, enabling it to be composed with other algorithms. Recently,  Frei, Gelles, Ghazy, and Nolin~\cite{content-oblivious-leader-election-24} answered this question for the special case of \emph{oriented rings}. They designed two asynchronous content-oblivious leader election algorithms, both with message complexity \( O(n \cdot \idmax) \), where \( n \) is the number of nodes and \( \idmax \) is the maximum identifier. The first algorithm \emph{stabilizes} in unoriented rings, without termination detection. The second algorithm \emph{quiescently} terminates in oriented rings, ensuring that no further pulses are sent to a node \( v \) after the algorithm on \( v \) declares termination.

The quiescent termination property enables the execution of the general simulation algorithm discussed earlier, following their leader election algorithm, as each node can correctly identify which pulses belong to which algorithm. As a result, this finding refutes the conjecture of Censor-Hillel, Cohen, Gelles, and Sela~\cite{fully-defective-22} for the special case of oriented ring topologies. Given the work by  Frei, Gelles, Ghazy, and Nolin~\cite{content-oblivious-leader-election-24}, the next question to address is whether content-oblivious leader election is feasible for general 2-edge-connected networks.

\begin{question}\label{Q2} In 2-edge-connected networks, is it possible to design a quiescently terminating content-oblivious leader election algorithm?
 \end{question}

Existing approaches for content-oblivious leader election heavily rely on the \emph{oriented}  ring structure, making  use of its two orientations to handle two types of messages. We briefly outline the key ideas behind the algorithms of  Frei, Gelles, Ghazy, and Nolin~\cite{content-oblivious-leader-election-24} as follows.

A central building block of their algorithms is a stabilizing leader election algorithm on oriented rings, which operates as follows. Each node generates a token, and all tokens travel in the same direction. Each node \( v \) maintains a counter to track the number of tokens it has passed. Once the counter reaches \( \id(v) \), node \( v \) destroys one token and temporarily elects itself as leader. Any incoming token will cause it to relinquish leadership. All counters eventually reach \( \idmax \) and stabilize. The node \( r \) with \( \id(r) = \idmax \) becomes the unique leader. This algorithm is not terminating, as nodes cannot determine whether additional tokens will arrive.

To achieve a quiescently terminating algorithm on oriented rings, they run the procedure in both directions sequentially, using the moment when the two counter values match as the termination condition. A node \( v \) begins the second run when its counter value from the first run reaches \( \id(v) \).

We stress that orientation is a key assumption of the above algorithm. In fact, they conjectured the impossibility of terminating leader election in unoriented rings~\cite{content-oblivious-leader-election-24}.

To achieve a stabilizing algorithm on unoriented rings, they run the procedure in both directions concurrently. The correctness follows from the observation that the same node is selected as the leader in both executions.

\subsection{Contributions}

We answer \Cref{Q1} affirmatively under the assumption that an upper bound \( N \geq n\) on the number of nodes $n=|V|$ is known, by presenting a content-oblivious leader election algorithm that quiescently terminates on any 2-edge-connected network.

  
\begin{restatable}{theorem}{mainthm}
    There is a quiescently terminating leader election algorithm with message complexity $O(m\cdot N\cdot \idmin)$ in any 2-edge-connected network $G=(V,E)$, where $m=|E|$ is the number of edges, $N \geq n$ is a known upper bound on the number of nodes $n=|V|$,  and $\idmin$ is the smallest $\id$. Moreover, the leader is the last node to terminate.
    \label{thm:main}
\end{restatable}

Combined with the simulation result of Censor-Hillel, Cohen, Gelles, and Sela~\cite{fully-defective-22}, our finding implies that in any 2-edge-connected network, any algorithm designed for the noiseless setting can be simulated in the fully-defective setting, without assuming a preselected leader. More precisely, as established in the prior work~\cite{content-oblivious-leader-election-24}, a sufficient condition for achieving this objective consists of (1) quiescent termination and (2) the requirement that the leader is the final node to terminate. This is because the algorithmic simulation in~\cite{fully-defective-22} is initiated by the leader.
Hence we answer \Cref{Q2} affirmatively, again under the assumption that an upper bound \( N \geq n\) is known. This refutes the original conjecture posed by Censor-Hillel, Cohen, Gelles, and Sela~\cite{fully-defective-22}.

For the special case of ring topologies, we can drop the assumption of knowing $N$ without requiring orientation.

\begin{restatable}{theorem}{secondmainthm}
There exists a quiescently terminating leader election algorithm with message complexity \(O(n \cdot \idmax)\) in any unoriented ring, where \(\idmax\) is the largest identifier. Moreover, the leader is the last node to terminate.
\label{thm:main2}
\end{restatable}

This result refutes the conjecture of~\citet{content-oblivious-leader-election-24}, demonstrating that orientation is not necessary for quiescently terminating leader election in rings.

\begin{table}[ht!]
\begin{center}
\caption{New and old results on content-oblivious leader election. }\label{main-table}
\begin{tabular}{lllll}
\textbf{Topology} & \textbf{Guarantee} & \textbf{Messages} & \textbf{Condition} & \textbf{Reference} \\ \hline
\multicolumn{1}{|l|}{Unoriented rings} & \multicolumn{1}{l|}{Stabilizing} & \multicolumn{1}{l|}{$O\left(n\cdot \idmax\right)$} & \multicolumn{1}{l|}{$\times$} & \multicolumn{1}{l|}{\cite{content-oblivious-leader-election-24}} \\ \hline
\multicolumn{1}{|l|}{Oriented rings} & \multicolumn{1}{l|}{Quiescently terminating} & \multicolumn{1}{l|}{$O\left(n\cdot \idmax\right)$} & \multicolumn{1}{l|}{$\times$} & \multicolumn{1}{l|}{\cite{content-oblivious-leader-election-24}} \\ \hline
\multicolumn{1}{|l|}{Oriented rings} & \multicolumn{1}{l|}{Stabilizing} & \multicolumn{1}{l|}{$\Omega\left(n \cdot \log \frac{\idmax}{n}\right)$} & \multicolumn{1}{l|}{$\times$} & \multicolumn{1}{l|}{\cite{content-oblivious-leader-election-24}} \\ \hline
\multicolumn{1}{|l|}{2-edge-connected} & \multicolumn{1}{l|}{Quiescently terminating} & \multicolumn{1}{l|}{$O\left(m\cdot N\cdot \idmin\right)$} & \multicolumn{1}{l|}{Known $N$} & \multicolumn{1}{l|}{\Cref{thm:main}} \\ \hline

\multicolumn{1}{|l|}{Unoriented rings} & \multicolumn{1}{l|}{Quiescently terminating} & \multicolumn{1}{l|}{$O\left(n\cdot \idmax\right)$} & \multicolumn{1}{l|}{$\times$} & \multicolumn{1}{l|}{\Cref{thm:main2}} \\ \hline

\end{tabular}
\end{center}
\end{table}

See \Cref{main-table} for a comparison of our results  and the results from prior work~\cite{content-oblivious-leader-election-24}. 

Our algorithm for general 2-edge-connected topologies has a higher message complexity of $O(m \cdot N \cdot \idmin)$ compared to the previous $O(n \cdot \idmax)$ bound in terms of network size, but it reduces the dependency on identifiers from $\idmax$ to $\idmin$. 

At first glance, the upper bound of $O(m \cdot N \cdot \idmin)$  may seem to contradict the lower bound of $\Omega\left(n \cdot \log \frac{\idmax}{n}\right)$ established in previous work~\cite{content-oblivious-leader-election-24}. 

However, the two results are consistent, as the lower bound does not make any assumptions about the value of $\idmin$. Precisely, their bound is of the form $\Omega\left(n \cdot \log \frac{k}{n}\right)$, where $k$ is the number of distinct, \emph{assignable} identifiers in the network. In other words, $k$ is the size of the $\id$ space.

While our algorithm for \Cref{thm:main} solves leader election in any 2-edge-connected topology, 

it does require a mild assumption: All nodes must \emph{a~priori} agree on a known upper bound $N$ on the number of nodes $n = |V|$ in the network, making it \emph{non-uniform}.  This assumption is widely adopted in many existing leader election and distributed graph algorithms, as upper bounds on network size are often available or can be estimated.  In contrast, the terminating algorithm from prior work~\cite{content-oblivious-leader-election-24} is \emph{uniform}, requiring no prior knowledge of the network, but it only works on oriented rings. 

Our leader election algorithm on rings requires neither orientation nor knowledge of $N$, while matching the message complexity of the algorithms of~\citet{content-oblivious-leader-election-24}. It therefore constitutes a \emph{strict} improvement over the previous state of the art. 

\subsection{Additional Related Work}

Fault-tolerant distributed computing involves two main challenges. The first is message corruption, which can occur due to factors such as channel noise. The second is the presence of faulty nodes and edges, which can arise from issues like unreliable hardware or malicious attacks.

To address message corruption, a common approach is to introduce additional redundancy using coding techniques, known as \emph{interactive coding}. The noise affecting the messages can be either random or adversarial, and there is usually an upper bound on the level of randomness or adversarial behavior. The study of interactive coding was initiated by Schulman~\cite{schulman1992communication,schulman1993deterministic,556671} in the two-party setting and later extended to the multi-party case by Rajagopalan and Schulman~\cite{rajagopalan1994coding}. For a comprehensive overview of interactive coding, see the survey by Gelles~\cite{gelles2017coding}.

 For networks with $n$ nodes, the maximum fraction of corrupted messages that a distributed protocol can tolerate is $\Theta(1/n)$~\cite{jain2015interactive}. If more than this fraction of messages are corrupted, the adversary can fully disrupt the communication of the node that communicates the fewest messages. Censor-Hillel, Gelles, and Haeupler~\cite{censor2019making} presented a distributed interactive coding scheme that simulates any asynchronous distributed protocol while tolerating an optimal corruption of $\Theta(1/n)$ of all messages. A key technique underlying their algorithm is a content-oblivious BFS algorithm.

The celebrated impossibility theorem of Fischer, Lynch, and Paterson~\cite{fischer1985impossibility} states that achieving consensus in an asynchronous distributed system is impossible for a deterministic algorithm when one or more nodes may crash.
 Similar impossibility results on the solvability of consensus have been found when processes are correct but may experience communication failures, such as message omissions, insertions, and corruptions \cite{santoro1989time}.

If the number of Byzantine edges is bounded by $f$, reliable communication can only be achieved if the graph is at least $2f$-edge-connected~\cite{dolev1982byzantine,pelc1992reliable}.

 A series of recent studies has focused on resilient and secure distributed graph algorithms in the {synchronous} setting~\cite{DBLP:conf/podc/ParterY19,DBLP:conf/soda/ParterY19a,DBLP:conf/wdag/HitronP21a,DBLP:conf/wdag/HitronP21,DBLP:conf/wdag/HitronPY22,DBLP:conf/podc/Parter22,DBLP:conf/podc/FischerP23,DBLP:conf/innovations/HitronPY23}. These works have developed compilation schemes that transform standard distributed algorithms into resilient and secure versions. For instance, in a $(2f+1)$-edge-connected network, any distributed algorithm in the $\mathsf{CONGEST}$ model can be adapted to remain resilient against up to $f$ adversarial edges~\cite{DBLP:conf/wdag/HitronP21a}.

 Content-oblivious computation has also been studied in the synchronous setting. In the \emph{beeping model} introduced by Cornejo and Kuhn~\cite{beeping-model-10}, during each communication round, a node can either beep or listen. If the node listens, it can distinguish between silence or the presence of at least one beep. Several leader election algorithms have been developed for the beeping model and its variations~\cite{beeping-multi-hop-radio-networks-19,beeping-time-optimal-18,forster2014deterministic,ghaffari2013near,beeping-minimal-weak-communication-25}.

Several additional distributed models have been designed with limited communication capabilities to capture various real-world constraints. These include radio networks~\cite{chlamtac2003broadcasting}, population protocols~\cite{angluin2006computation}, and stone-age distributed computing~\cite{emek2013stone}. Leader election has been extensively studied in these models~\cite{alistarh2022near, berenbrink2020optimal, beeping-multi-hop-radio-networks-19, czumaj2021exploiting, ChangKPWZ17, doty2018stable, emek2021thin,  ghaffari2013near, sudo2020leader, beeping-minimal-weak-communication-25}.
\section{Preliminaries}

We present the model and problem considered in the paper and overview some basic terminology. 

\subsection{Content-Oblivious Model} The communication network is abstracted as a graph $G = (V,E)$, where each node $v \in V$ is a computing device and each edge $e \in E$ is a bidirectional communication link. We write $n = |V|$ and $m = |E|$. Each node $v$ has a unique identifier $\id(v)$, and we define $\idmax = \max_{v \in V} \id(v)$ and $\idmin = \min_{v \in V} \id(v)$. Throughout the paper, we assume that $G$ is 2-edge-connected, meaning that the removal of any single edge does not disconnect the graph.

\paragraph{General 2-edge-connected topologies.}
 For our result on general 2-edge-connected topologies,  we assume no prior knowledge of the network topology, except that each node is given an upper bound $N \geq n$ on the total number of nodes. That is, the algorithm is non-uniform.

The degree of a node $v$, denoted by $\deg(v)$, is the number of edges incident to $v$. Although a node does not initially know the identifiers of its neighbors, it can distinguish among its incident edges using \emph{port numbering}---a bijection between the set of incident edges and the set $\{1, 2, \ldots, \deg(v)\}$. Therefore, we use the term port to refer to the local endpoint of an edge, uniquely identified by its port number, which is used to send and receive messages to and from a specific neighbour.

\paragraph{Unoriented rings.}
For our result on rings, we assume that the topology is an unoriented ring.
In this case, each node $v_i$ is connected to nodes $v_{i-1}$ and $v_{i+1}$ (indices taken modulo $n$) by two local communication ports: port $0$ and port $1$.   
These local labels do not induce a consistent orientation of the ring; that is, the global assignment of labels $0$ and $1$ to ports is arbitrarily set at the beginning by an adversary. It may be the case that a message sent by $v_1$ on port $0$ is delivered to node $v_0$ on its local port $1$, and a message sent by $v_1$ on port $1$ is delivered to $v_2$ on its local port $1$. A correct algorithm must work for any possible assignment of port labels.

\paragraph{Content-oblivious communication.}
In the content-oblivious setting, nodes communicate by sending \emph{pulses}, which are messages without any content. We consider the \emph{asynchronous} model, where node behavior is \emph{event-driven}: A node can act only upon initialization or upon receiving a pulse. Based on the pattern of previously received pulses, each node decides its actions, which may include sending any number of pulses through any subset of its ports. We emphasize that when a node receives a pulse, it knows the port through which the pulse arrived. Likewise, a node can specify how many pulses to send through each port. We restrict our attention to the \emph{deterministic} setting, in which nodes do not use randomness in their decision-making.

\paragraph{Time.} There is no upper bound on how long a pulse may take to traverse an edge, but every pulse is guaranteed to be delivered after a finite delay. An equivalent way to model this behavior is by assuming the presence of a scheduler. At each time step, as long as there are pulses in transit, the scheduler arbitrarily selects one and delivers it.

We assume that upon receiving a pulse, all actions triggered by that pulse are performed \emph{instantaneously}, so that the system can be analyzed in \emph{discrete} time steps. We focus only on specific, well-defined moments in the execution of the algorithm---namely, the moments just before the scheduler selects a pulse to deliver. At these times, all the actions resulting from earlier pulse deliveries have been completed, and the system is in a stable state, awaiting the next step. We do not analyze the system during any intermediate period in which a node is still performing the actions triggered by a pulse arrival.

\paragraph{Counting the number of pulses.}
We write $[s]=\{1,2,\ldots,s\}$.
In the algorithm description and analysis for \Cref{thm:main}, for each node $v \in V$ and each port number $i \in [\deg(v)]$, we use the following notation to track the number of pulses sent and received so far:

\begin{itemize}
\item $\sigma_i(v)$ denotes the number of pulses sent on port $i$ of node $v$.
\item $\rho_i(v)$ denotes the number of pulses received on port $i$ of node $v$.
\end{itemize}

\subsection{Leader Election}
We now define the \emph{leader election} problem.

\begin{definition} [Leader election] \label{def:LE}
The task of leader election requires that each node outputs exactly one of $\leader$ and $\nonleader$, with the additional requirement that exactly one node outputs $\leader$.
\end{definition}

More precisely, each node maintains a binary internal variable $\isleader$ to represent its current output, taking the value $\leader$ or $\nonleader$. The value of $\isleader$ may change arbitrarily many times during the execution of the algorithm, but it must eventually stabilize to a final value that remains unchanged thereafter. The final outputs of all nodes must collectively satisfy the conditions of \Cref{def:LE}.

Due to the asynchronous nature of our model, there are various ways in which an algorithm can \emph{finish}. Based on this, we classify algorithms into three types, ordered by increasing strength of termination guarantees. 

   \begin{definition} [Stabilization]
 A node stabilizes at time $t$ if its output does not change after time $t$. An algorithm is {stabilizing} if it guarantees that all nodes eventually stabilize.
    \end{definition}

    \begin{definition} [Termination]
    A node terminates at time $t$ if it explicitly declares termination at time $t$. After declaring termination, the node ignores all incoming pulses and can no longer change its output. An algorithm is {terminating} if it guarantees that all nodes eventually terminate.
    \end{definition}

    \begin{definition} [Quiescent termination]
    A node quiescently terminates at time $t$ if it terminates at time $t$ and no further pulses arrive at the node thereafter. An algorithm is {quiescently terminating} if it guarantees that all nodes eventually quiescently terminate.
    \end{definition}

Ideally, we aim to design leader election algorithms that are quiescently terminating, as this property allows a second algorithm to be executed safely after leader election. In contrast, if the algorithm guarantees only stabilization or termination, nodes may be unable to distinguish whether an incoming pulse belongs to the leader election process or the subsequent algorithm. This ambiguity can affect the correctness of both algorithms. For this reason, prior work~\cite{content-oblivious-leader-election-24} establishes a general algorithm simulation result only for oriented rings, but not for unoriented rings.

\subsection{DFS and Strongly Connected Orientation}\label{sect:dfs_orientation}

We describe a DFS tree $T_G$ of the graph $G = (V, E)$ that is \emph{uniquely determined} by the port numbering and node identifiers. Based on this DFS tree, we describe an edge orientation of $G$ that transforms it into a strongly connected directed graph $\vec{G}$. Both the DFS tree and the orientation play a critical role in the analysis of our algorithm for \Cref{thm:main}.

\paragraph{DFS.} 
Throughout this paper, we write $r$ to denote the node with the smallest identifier.
Consider a DFS traversal rooted at $r$ such that when a node selects the next edge to explore among its unexplored incident edges, it chooses the one with the \emph{smallest} port number. Let $T_G = (V, E_T)$ be the resulting DFS tree. The edges in $E_T$ are called \emph{tree edges}, and the remaining edges $E_B = E \setminus E_T$ are called \emph{back edges}. The following observation is folklore.

\begin{observation} \label{obs: back edge}
    For any edge $\{u, v\} \in E_B$, either $u$ is an ancestor of $v$, or $v$ is an ancestor of $u$.
\end{observation}
\begin{proof}
When a node $u$ explores an incident edge $e = \{u, v\}$ during the DFS, there are two possibilities: Either $v$ is an ancestor of $u$, or $v$ has not been visited yet. In the first case, $\{u, v\} \in E_B$. In the second case, $v$ is discovered for the first time, so $\{u, v\} \in E_T$ is added as a tree edge in $E_T$. 
\end{proof}

\begin{definition}[Strongly connected orientation]\label{def-directed}
    The directed graph $\vec{G} = (V, \vec{E}_T \cup \vec{E}_B)$ is defined by
\begin{align*}
   \vec{E}_T &= \{ (u,v) \, | \, \{u,v\} \in E_T \text{ and } u \text{ is the parent of } v \text{ in } T_G\},\\
   \vec{E}_B &= \{(u,v) \, | \, \{u,v\} \in E_B \text{ and } v \text{ is an ancestor of } u \text{ in } T_G\}.
\end{align*}
\end{definition}

A directed graph is \emph{strongly connected} if for any two nodes $u$ and $v$ in the graph, there is a directed path from $u$ to $v$.
We have the following observation.

\begin{observation} \label{obs:sc}
    If $G$ is 2-edge-connected, then $\vec{G}$ is strongly connected.
\end{observation}
\begin{proof}
To prove the observation, it suffices to show that, for any tree edge $e=\{u,v\}$, both $u$ and $v$ are in the same strongly connected component of $\vec{G}$, as this implies that all nodes in $V$ are in the same strongly connected component. Let $u$ be the parent of $v$. By \Cref{def-directed}, we already know that $(u,v)$ is a directed edge of $\vec{G}$, so it remains to show that $u$ is reachable from $v$ in $\vec{G}$.

Let $S$ be the set consisting of $v$ and all descendants of $v$ in the DFS tree $T_G$. Since $G = (V,E)$ is 2-edge-connected, there exists an edge $e' \in E \setminus \{e\}$ that connects $V\setminus S$ and $S$. Since $e$ is the only tree edge connecting $V\setminus S$ and $S$, $e'$ must be a back edge. Let $e'= \{x,y\}$ where $y$ is an ancestor of $x$ in $T_G$.

Since $x \in S$, either $x=v$ or $x$ is a descendant of $v$, so there is a directed path $P_{v \rightarrow x}$ from $v$ to $x$ using only edges in $\vec{E}_T$. Since $y \in V \setminus S$ is an ancestor of $x$, either $y=u$ or $y$ is an ancestor of $u$, so there is a directed path $P_{y \rightarrow u}$ from $y$ to $u$ using only edges in $\vec{E}_T$. Combining $P_{v \rightarrow x}$, $(x,y)$, and $P_{y \rightarrow u}$, we infer that $u$ is reachable from $v$ in $\vec{G}$.
\end{proof} 

\paragraph{Shortest paths.} Throughout the paper, for any two nodes $u \in V$ and $v \in V$, we write $P_{u,v}$ to denote any shortest path from $u$ to $v$ in $\vec{G}$, whose existence is guaranteed by \Cref{obs:sc}. We allow $u=v$, in which case $P_{u,v}$ is a path of zero length consisting of a single node $u=v$.

\section{Technical Overview}

We now give a high-level overview of the proofs for our \Cref{thm:main,thm:main2}. We first consider general 2-edge-connected topologies, and then turn to unoriented rings.

\subsection{Overview of \texorpdfstring{\Cref{thm:main}}{Theorem 1}}
We present a technical overview for the proof of \Cref{thm:main}.   Our algorithm aims to elect the node $r$ with the smallest identifier as the leader. It operates in two phases. In the first phase, each node maintains a counter that is incremented in an almost synchronized manner using the \emph{$\alpha$-synchronizer} of Awerbuch~\cite{awerbuch1985complexity}. A node $v$ declares itself the leader once its counter reaches $\id(v)$. This first phase alone yields a simple \emph{stabilizing} leader election. To ensure \emph{quiescent termination}, the second phase is executed, in which the elected leader announces its leadership through a DFS traversal.

\subsubsection{Synchronized Counting}

We begin by describing the first phase of the algorithm.

\paragraph{Counting.}  
To achieve an almost synchronized counting, we employ the $\alpha$-synchronizer of~\citet{awerbuch1985complexity}, as follows.
Each node maintains a counter variable $\counter(v)$, which is initialized to $0$ at the start of the algorithm. A node increments its counter by one after receiving a new pulse from each of its neighbors. Formally, the counter value of a node $v$ is given by  
\[
\counter(v) = \min_{j \in [\deg(v)]} \rho_j(v).
\]  
Each node broadcasts a pulse to all its neighbors upon initialization and immediately after each counter increment. 

This process is \emph{non-blocking}: In the absence of an exit condition, nodes will continue counting together indefinitely. Furthermore, no node can advance significantly faster than its neighbors: For any pair of adjacent nodes, their counter values differ by at most one. A faster-counting node must wait for a pulse from a slower-counting neighbor before it can increment its counter again.

\paragraph{Stabilizing leader election.}  
A node $v$ elects itself as the leader once  
\[
\counter(v) = \id(v).
\]  
After electing itself as the leader, node $v$ \emph{freezes} its counter at $\counter(v) = \id(v)$ and does not increment it further. This action caps the counter value of any neighbor of $v$ at most $\id(v) + 1$. More generally, in a graph of $n \leq N$ nodes, the counter value of any node is upper bounded by $\id(v) + N - 1$, since any two nodes are connected by a path of length at most $N - 1$.

To ensure that the algorithm elects the node $r$ with the smallest identifier as the \emph{unique} leader, it suffices for any two identifiers to differ by at least $N$. This can be achieved by first multiplying each identifier by $N$. This is where the \emph{a priori} knowledge of $N$ comes into play.

At this stage, we have a simple stabilizing leader election algorithm. However, it remains insufficient for our purposes: We also want all non-leader nodes to recognize their non-leader status and terminate in a quiescent manner.

\subsubsection{DFS Notification}

We proceed to describe the second phase of the algorithm. The leader must find a way to notify all other nodes to settle as non-leaders. This seems challenging because the leader can only send additional pulses, which appear indistinguishable from those in the first phase. As we will see, assuming the network is 2-edge-connected, there is a simple and effective approach: Just send many pulses to overwhelm the listener.

\paragraph{Notifying a neighbor.} Let \( u \) be a neighbor of the leader \( r \), and let the edge \( e = \{u, r\} \) correspond to port \( i \) at \( r \) and port \( j \) at \( u \). The leader \( r \) notifies \( u \) of its leadership as follows:
\begin{enumerate}
    \item Wait until \( \rho_i(r) \geq \id(r) + 1 \).
    \item Send pulses to \( u \) until \( \sigma_i(r) = \id(r) + N + 2 \).
\end{enumerate}

Since the graph \( G \) is 2-edge-connected, there exists a path \( P \) from \( u \) to \( r \) that avoids \( e \), implying
\[
\counter(u) \leq \counter(r) + \length(P) \leq \id(r) + N - 1.
\]
Eventually, \( u \) receives \( \id(r) + N + 2 \) pulses from \( r \), so there must be a moment where
\[
\rho_j(u) \geq \counter(u) + 3,
\]
triggering anomaly detection. Under normal conditions in the first phase, \( u \) would expect
\[
\rho_j(u) \leq \sigma_i(r) = \counter(r) + 1 \leq \counter(u) + 2,
\]
so this deviation signals a violation.

Additionally, the waiting condition \( \rho_i(r) \geq \id(r) + 1 \) ensures that \( \counter(u) \geq \id(r) \). Thus, \( u \) not only recognizes that it is not the leader but can also deduce the exact value of \( \id(r) \) using
\[
\id(r) = N \cdot \left\lfloor \frac{\counter(u)}{N} \right\rfloor,
\]
since all node identifiers are integer multiples of \( N \).

\paragraph{DFS.} Intuitively, the leader $r$ is able to notify a neighbor $u$ because the counter value of node $u$ is constrained by a \emph{chain} $P$, whose other end is \emph{anchored} at the leader $r$, whose counter has already been frozen. To extend this notification process to the entire network, we perform a DFS, using the same notification procedure to explore each new edge. When a node finishes exploring all its remaining incident edges, it can notify its parent by sending pulses until a total of $\id(r)+N+2$ pulses have been sent through this port since the start of the algorithm. When a node $u$ explores a back edge $\{u,v\}$, the node $v$ replies using the same method. When a node decides which new incident edge to explore, it prioritizes the edge with the smallest port number, so the resulting DFS tree is $T_G$, as defined in \Cref{sect:dfs_orientation}. 

\paragraph{Quiescent termination.} To see that the algorithm achieves quiescent termination, observe that by the end of the execution, for each edge~$e$, the number of pulses sent in both directions is exactly $\id(r) + N + 2$. Therefore, once a node observes that it has both sent and received exactly $\id(r) + N + 2$ pulses along each of its ports and has completed all local computations, it can safely terminate with a quiescent guarantee.

\paragraph{Correctness.} A similar chain-and-anchor argument can be applied to analyze the correctness of the DFS-based notification algorithm. We prove by induction that for each node $u$ that has been visited by DFS, its frozen counter value satisfies
\[\counter(u) \leq \id(r) + \length(P_{u,r}).\]
Recall from \Cref{sect:dfs_orientation} that $P_{a,b}$ is a shortest path from $a$ to $b$ in the directed version $\vec{G}$ of $G$ defined in \Cref{def-directed}.

Suppose during the DFS, a node \( x \) is about to explore an incident edge \( e = \{x, y\} \) using the notification procedure, and that \( y \) has not yet been visited. Let \( w \) be the first node on the path \( P_{y,r} \) that is either \( x \) or an ancestor of \( x \) in the DFS tree \( T_G \), then we have
\[
\length(P_{y,r}) = \length(P_{y,w}) + \length(P_{w,r}).
\]
We now apply the chain-and-anchor argument using the chain \( P_{y,w} \) and anchor \( w \). This choice of chain is valid because the path \( P_{y,w} \) cannot include the edge \( e \), as \( e \) is oriented from \( x \) to \( y \) in \( \vec{G} \). This is precisely why the analysis is carried out in \( \vec{G} \) rather than in \( G \).

By the induction hypothesis, the frozen counter value at the anchor satisfies $
\counter(w) \leq \id(r) + \length(P_{w,r})$.
It follows that
\begin{align*}
\counter(y) &\leq \counter(w) + \length(P_{y,w}) \\
&\leq \id(r) + \length(P_{w,r}) + \length(P_{y,w}) \\
&= \id(r) + \length(P_{y,r}) \\
&\leq \id(r) + N - 1.
\end{align*}
Therefore, the chain-and-anchor argument works, allowing \( y \) to be successfully notified. Moreover, the frozen counter value of \( y \) also satisfies the induction hypothesis.

In addition, we must ensure that the DFS procedure proceeds correctly and no event is triggered unexpectedly. Once again, the chain-and-anchor argument implies that the condition
\[
\counter(u) \leq \id(r) + \length(P_{u,r}) \leq \id(r) + N - 1<\id(u)
\]
holds for all nodes \( u \in V\setminus\{r\}\) at all times. This ensures that no node other than \( r \) can be mistakenly elected as a leader. It also guarantees that pulses sent by nodes still in the first phase cannot accidentally trigger a notification in the second phase. Specifically, if $u$ is still in the first phase, then for every port \( j \in [\deg(u)] \), we have
\[
\sigma_j(u) = \counter(u) + 1 \leq \id(r) + N,
\]
which remains below the threshold of \( \id(r) + N + 2 \) pulses required to confirm receipt of a notification.

\paragraph{Remark.} One might wonder why the notification must follow a DFS rather than allowing a node to notify all neighbors simultaneously to accelerate the process. A key reason is that doing so can break the chain-and-anchor argument by eliminating anchors. For instance, in a ring topology, if the node $r$ with the smallest identifier simultaneously sends multiple pulses to \emph{both} neighbors, all of the pulses might be interpreted as part of the first phase, with no detectable anomaly, regardless of how many pulses are sent.

Both our leader election algorithm and the simulation algorithm of Censor-Hillel, Cohen, Gelles, and Sela~\cite{fully-defective-22} utilize DFS, but for different purposes. They perform a DFS from a pre-selected leader to find a cycle, while we use DFS to disseminate the identity of the elected leader to all nodes and to achieve quiescent termination. Coincidentally, both approaches leverage the strongly connected orientation of 2-edge-connected graphs.  In our case, the orientation $\vec{G}$ is used solely in the analysis to facilitate the chain-and-anchor argument. In their case, it guarantees the existence of a Robbins cycle, which plays a critical role in their algorithm. 

\subsection{Overview of \texorpdfstring{\Cref{thm:main2}}{Theorem 2}}

Our uniform algorithm for rings is divided into phases and is based on the assumption that no two identifiers are consecutive.  To ensure this, during initialization we multiply every identifier by two. 

The first phase of the algorithm uses a synchronization mechanism similar to the one used for general 2-edge-connected topologies. Each node sends a pulse to each of its neighbors and waits for pulses from them. This process is repeated a number of times equal to the node’s identifier.

Notice that the chain-and-anchor argument of the previous section shows that the process with maximum $\id$ cannot exit the synchronization phase before the process with minimum $\id$.

Once a node terminates the synchronization phase (called the \emph{competition phase} in this algorithm), it enters a \emph{solitude checking phase} to test whether its neighbours are still competing. This is done by sending two pulses on port $0$ and waiting to receive two pulses. If the two pulses received are in the opposite direction from the ones sent (i.e., on port $1$), the node terminates in the leader state. Otherwise, the node becomes a non-leader. A non-leader first sends two pulses on port $1$ to ``balance'' the number of pulses sent in each direction, and then, after receiving a pulse from each direction, enters a phase in which it simply relays pulses.

By doing so, a non-leader node behaves in such a way that it is transparent to the other nodes.

The node with maximum $\id$ will be the last to leave the competition phase, with all other nodes acting as relays. Thus, its solitude check will succeed, and the leader will notify all others by sending an additional pulse on port $0$. This new pulse, together with the two pulses used for the leader’s solitude check, will form a group of three pulses traveling in the same direction around the ring. This group can be detected by all non-leader processes, by doing the difference on the number of pulses received from each port, which will then terminate.

\paragraph{High-level proof.}
Consider the behaviour of the node $v$ with the minimum identifier $\idmin$, and its two neighbours $u$ and $w$. Once $v$ starts its solitude checking phase, both $u$ and $w$ are still in the competition phase, waiting for pulses from $v$. In the solitude checking phase, $v$ sends two pulses on port $0$ (to node $w$) and waits for pulses. Since $u$ is still in its competition phase, it will send just one pulse to $v$ on port $1$ and will wait for a response from $v$ before sending any further pulses. Therefore, the next pulse $v$ receives must come from $w$ on port $0$, causing $v$ to fail the solitude check.

Node $v$ will then send two pulses on port $1$ (balancing phase). We can show that at this point, both $u$ and $w$ have received from $v$ exactly $\idmin+2$ pulses, in an execution indistinguishable from one on a smaller ring where $v$ does not exist and $u$ and $w$ are directly connected. Afterwards, $v$ will only relay pulses, acting as an asynchronous link between $u$ and $w$.

Our proof idea builds on the argument above: all nodes except the one with maximum $\id$ will eventually reach the solitude check and fail. In doing so, they behave in a way that hides their presence from their neighbours. This allows us to inductively reason on progressively smaller rings until only the node with maximum $\id$ remains.

\section{Non-Uniform Leader Election for 2-Edge-Connected Topologies}\label{sec:qle2ed}

In this section, we prove our main theorem.

\mainthm*

The algorithm for \Cref{thm:main} is presented in \Cref{alg:2-connected}, with the DFS notification subroutine detailed in \Cref{alg:notify}. There are five main variables used in our algorithm:
\begin{description}
    \item[{\boldmath $\state(v)$}:]  This variable represents the current state of $v$.  We say a node $v$ is in the \emph{synchronized counting} phase if $\state(v) = \bot$. Otherwise, the node is in the \emph{DFS notification} phase, and $\state(v)\in\{\leader,\nonleader\}$ indicates the output of $v$. At the start, every node is in the synchronized counting phase. Once a node advances to the DFS notification phase, it fixes its output and cannot go back to the synchronized counting phase. 
\item [{\boldmath$\leader\id(v)$}:] This variable stores the identifier of the leader as known by node $v$.
\item [{\boldmath$\parentport(v)$}:] This variable stores the port number through which $v$ connects to its parent in $T_G$.
\item [{\boldmath$\counter(v)$}:] This variable implements the counter used in the synchronized counting phase.
\item [{\boldmath$\mathcal{P}(v)$}:] This variable represents the remaining unexplored ports of $v$ in the DFS notification phase. Precisely, $\mathcal{P}(v)$ is the set of ports on which $v$ has not yet received $\leader\id+N+2$ pulses.
\end{description}

For ease of presentation, we define the following two actions:

\begin{description}
    \item[{\boldmath $\sendpulsesuntil_i(k)$}:] Send pulses along port $i$ until a total of $k$ pulses have been sent from the start of the algorithm (i.e., $\sigma_i=k$).
    \item[{\boldmath $\sendall(1)$}:]  Send one pulse along each port.
\end{description}

\paragraph{Events.} For ease of analysis, several positions in \Cref{alg:2-connected} and \Cref{alg:notify} are annotated with labels: $\start(v)$ (node~$v$ initiates the DFS), $\sendexplore_i(v)$ (node~$v$ begins sending an explore-notification on port~$i$), $\receiveexplore_i(v)$ (node~$v$ confirms receipt of an explore-notification from port~$i$), $\senddone_i(v)$ (node~$v$ begins sending a done-notification on port~$i$), and $\receivedone_i(v)$ (node~$v$ confirms receipt of a done-notification from port~$i$). These labels are treated as \emph{events} in the analysis.

\paragraph{Pulse labeling.}  Any pulse that is sent during \Cref{alg:notify} is called a \emph{DFS notification} pulse. 
Any remaining pulse is called a \emph{synchronized counting} pulse.
This distinction is used only in the analysis. The nodes cannot see these labels in the algorithm.

\begin{algorithm}[ht!]
\DontPrintSemicolon
\caption{Leader election algorithm for node $v$}
\label{alg:2-connected}
$\id(v) \gets \id(v) \cdot N$ \tcp*{Forcing all $\id$s to be integer multiples of $N$}
$\state \gets \bot$;
$\leader\id \gets \bot$;
$\parentport \gets \bot$\;
$\counter \gets 0$\;
$\sendall(1)$\;

\While(\tcp*[f]{Synchronized counting}){$\state = \bot$} { 
    \OnEvent{$\rho_i$ is incremented for some $i \in [\deg(v)]$} {
        \uIf(\tcp*[f]{Counter increment condition}){$\rho_i$ = $\min_{j \in [\deg(v)]} \rho_j$} {
            $\counter \gets \counter + 1$\;
            $\sendall(1)$\;
            \uIf(\tcp*[f]{$\leader$ exit condition}) {$\counter = \id(v)$} {
                $\state \gets \leader$\;
                $\leader\id \gets \id(v)$\;
                \textcolor{red}{Event: $\start(v)$}\;
            }
        }
        \uElseIf(\tcp*[f]{$\nonleader$ exit condition}) {$\rho_i - \sigma_i > 1$} {
            $\state \gets \nonleader$\;
            $\parentport \gets i$\;
            $\leader\id \gets \lfloor \counter / N \rfloor \cdot N$\;
            \textbf{wait until} $\rho_i = \leader\id + N+2$\;
            \textcolor{red}{Event: $\receiveexplore_i(v)$}\;
        }
    }
} 

$\notify(\state,\leader\id,\parentport)$ \tcp*{\Cref{alg:notify}}

\end{algorithm}

\begin{algorithm}[ht!]
\DontPrintSemicolon
\caption{$\notify(\state,\leader\id,\parentport)$ for node $v$}
\label{alg:notify}

$\mathcal{P} = [\deg(v)]$ \tcp*[r]{The set of all ports}
\uIf{$\state = \nonleader$}{
    $\mathcal{P} \gets \mathcal{P} \setminus\{\parentport\}$ \;
}

\While {$\mathcal{P} \neq \emptyset$} {
    $j \leftarrow \min(\mathcal{P})$ \tcp*[r]{Traversal order: prioritizing smaller port numbers}
    \textcolor{red}{Event: $\sendexplore_j(v)$} \;
    \textbf{wait until} {$\rho_j \geq \leader\id+1$} \;
    $\sendpulsesuntil_{j}(\leader\id + N+2)$ \tcp*[r]{Continue DFS on port $j$}
    \While {$j \in \mathcal{P}$} {
        \OnEvent{$\rho_h = \leader\id+N+2$ for some $h \in \mathcal{P}$}{
            \uIf(\tcp*[f]{DFS returned from port $j$}){$j=h$} {
                $\mathcal{P} \gets \mathcal{P}\setminus\{j\}$ \;
                \textcolor{red}{Event: $\receivedone_j(v)$} \;
            }
            \uElse(\tcp*[f]{Incoming DFS from port $h$ via a back edge}){
                $\mathcal{P} \gets \mathcal{P}\setminus\{h\}$ \;
                \textcolor{red}{Event: $\receiveexplore_h(v)$} \;
                \textcolor{red}{Event: $\senddone_h(v)$} \;
                $\sendpulsesuntil_{h}(\leader\id + N+2)$ \;
            }
        }
    } 
}
\uIf{$\state = \nonleader$}{
    \textcolor{red}{Event: $\senddone_{\parentport}(v)$} \;
    $\sendpulsesuntil_{\parentport}(\leader\id + N+2)$ \tcp*[r]{Return DFS on parent port}
}
\end{algorithm}

\subsection{Synchronized Counting}\label{subsect:counting}

We begin by analyzing the synchronized counting phase of our algorithm, which corresponds to the main loop of \Cref{alg:2-connected}. We aim to show that the node $r$ with the smallest identifier eventually elects itself as the leader ($\state = \leader$) while all other nodes are still in the synchronized counting phase ($\state=\bot$).

\Cref{alg:2-connected} starts with $\counter=0$. The action $\sendall(1)$ is performed upon initialization and immediately after each counter increment, so $\sigma_i(v) = \counter(v)+1$ for all $i \in [\deg(v)]$. The counter is incremented by one whenever an incoming pulse causes $\min_{j \in [\deg(v)]} \rho_j$ to increase, so we always have $\counter = \min_{j \in [\deg(v)]} \rho_j$. There are two ways to exit the main loop of \Cref{alg:2-connected}. The first is $\counter = \id(v)$, which causes $v$ to become a leader. The second is $\rho_i-\sigma_i>1$, which causes $v$ to become a non-leader. 

Within \Cref{subsect:counting}, we focus on analyzing the behavior of the nodes in the synchronized counting phase. The only fact that we need to know about the DFS notification phase is that upon exiting the main loop of \Cref{alg:2-connected}, the value of $\counter$ is \emph{permanently fixed}. The analysis of the DFS notification phase is done in \Cref{subsect:DFS}.



\begin{observation}[Outgoing synchronized counting pulses] \label{obs1-variant} \label{obs:s=c+1}
For any node $v \in V$, at any time in the execution of the algorithm, the number of synchronized counting pulses sent on each of its ports equals $\counter(v)+1$.
\end{observation}
\begin{proof}
During the synchronized counting phase, the action $\sendall(1)$ is executed during initialization and immediately after each counter increment. There is no other mechanism by which node~$v$ sends synchronized counting pulses in the algorithm.
Since $\counter(v)$ starts at zero and becomes fixed when $v$ exits the main loop of \Cref{alg:2-connected}, it follows that for any node $v \in V$, at any time during the algorithm, the number of synchronized counting pulses sent on each of its ports equals $\counter(v)+1$.
\end{proof}



The above observation implies that for any node $v \in V$,  while $v$ is still in the synchronized counting phase,  $\sigma_i(v) = \counter(v)+1$ for all $i \in [\deg(v)]$.

The following observation relates the counter value with the number of pulses sent on each port. 

\begin{observation}[Counter value] \label{obs:r>=c}
 For any node $v \in V$,  while $v$ is in the synchronized counting phase, $\counter(v) = \min_{j \in [\deg(v)]} \rho_j(v)$.
\end{observation}

\begin{proof}
Since $\counter(v)$ starts at zero, it suffices to show that, while $v$ is in the synchronized counting phase, the counter is incremented by one whenever an incoming pulse causes $\min_{j \in [\deg(v)]} \rho_j$ to increase, and that there is no other mechanism for changing the counter value. This follows from the observation that a pulse arriving on port~$i$ increases $\min_{j \in [\deg(v)]} \rho_j$ if and only if $\rho_i = \min_{j \in [\deg(v)]} \rho_j$ immediately after the pulse arrives. This is precisely the condition for incrementing the counter in \Cref{alg:2-connected}.
\end{proof}

Since the counter is frozen once a node exits the synchronized counting phase, there is a simple generalization of \Cref{obs:r>=c} to any node $v$ in any phase: $\counter(v) \leq \min_{j \in [\deg(v)]} \rho_j(v)$.

The following lemma shows that the counting process is nearly synchronized along every edge, as long as it has not been disturbed by a DFS notification pulse.


 \begin{lemma}[Near synchronization] \label{lem: c differ by <=1}
    For any edge $e = \{u, v\} \in E$ with no DFS notification pulse delivered yet,  \[|\counter(u) - \counter(v)| \leq 1.\]
 \end{lemma}

 \begin{proof}
 Suppose the port number of $e$ is $j$ at node~$u$ and $i$ at node~$v$. Since the counter is frozen once a node exits the synchronized counting phase,
   \Cref{obs:r>=c} implies that  $\counter(u) \leq \rho_j(u)$. In other words, the number of pulses from $v$ that have been delivered to $u$ is at least $\counter(u)$. \Cref{obs1-variant}, together with the assumption that no DFS notification pulse from $v$ has been delivered to $u$, implies that $u$ receives at most $\counter(v) + 1$ pulses from $v$, so $\rho_j(u) \leq \counter(v) + 1$. Combining the two inequalities, we obtain $\counter(u) \leq  \counter(v)+1$. Similarly, we also have $\counter(v) \leq \counter(u)+1$, so $|\counter(u) - \counter(v)| \leq 1$.
 \end{proof}

In the following lemma, we show that the $\nonleader$ exit condition must be triggered by a neighbor already in the DFS notification phase, meaning the first node to exit the synchronized counting phase must do so via the $\leader$ exit condition.

 \begin{lemma}[$\nonleader$ exit condition] \label{lem: rho-sigma<=1}
Suppose the $\nonleader$ exit condition
 \[\rho_i(v)-\sigma_i(v) > 1\]
is triggered by a pulse arriving on port $i$ at node $v$, then the pulse must be a DFS notification pulse.
 \end{lemma} 

 \begin{proof}
 Let $u$ be the neighbor of $v$ such that the port number of $e = \{u, v\}$ at node $v$ is $i$. 
While $v$ is in the synchronized counting phase, by \Cref{obs:s=c+1}, $\sigma_i(v) = \counter(v)+1$. If $v$ has not received any DFS notification pulse from $u$, then \Cref{obs1-variant} and \Cref{lem: c differ by <=1} imply that the number of pulses $v$ received from $u$ is $\rho_i(v) \leq \counter(u)+1 \leq \counter(v)+2$, so
\[\rho_i(v)-\sigma_i(v) \leq 1,\]
which is insufficient to trigger the  $\nonleader$ exit condition for $v$.
 \end{proof}

Intuitively, the following lemma shows that deadlock cannot occur during the counting process: Within a finite amount of time, the system always makes progress in at least one of the following three ways. We use the convention that $\min X = \infty$ if $X = \emptyset$. Therefore, for the special case of $S=V$, (C1) in the  lemma below becomes $\min_{v \in S} \counter(v) = \infty$, meaning that only (C2) or (C3) hold.

\begin{lemma}[Deadlock freedom]\label{lem: no deadlock}
Let $S\subseteq V$ be the set of all nodes currently in the synchronized counting phase. 
    Suppose no DFS notification pulses have yet been delivered to any node in $S$, then eventually at least one of the following conditions hold.
    \begin{description}
        \item[(C1)]   $1 + \min_{u \in V \setminus S} \counter(u) \leq \min_{u \in S} \counter(u)$.
        \item[(C2)]  Some node in $S$ exits the synchronized counting phase. 
        \item[(C3)] A DFS notification pulse from some node in $V\setminus S$ is delivered to some node in $S$.
    \end{description}
\end{lemma}
\begin{proof}
Let $c = \min_{u\in S} \counter(u)$. Select $v$ as an arbitrary node with $\counter(v)=c$. We claim that, within finite time, at least one of the following events occurs: (C1), (C2), (C3), or $\counter(v)\geq c+1$. Therefore, by applying the claim $x \cdot |S|$ times, it follows that either one of (C1), (C2), or (C3) happens, or the value of $\min_{v\in S} \counter(v)$ increases by at least $x$. In the latter case, by choosing  $x = \min_{u \in V \setminus S} \counter(u)$, we ensure that (C1) happens.

For the rest of the discussion, we prove the claim. For each neighbor $u$ of $v$ in $S$, we have $\counter(u) \geq c$. By \Cref{obs:s=c+1}, $u$ has already sent $\counter(u)+1 \geq c+1$ pulses to $v$. If (C1) has not happened yet, for each neighbor $w$ of $v$ in $V \setminus S$, we also have $\counter(w) \geq c$. By \Cref{obs1-variant}, $w$ has already sent $\counter(w)+1 \geq c+1$ pulses to $v$.

Upon the arrival of these pulses, we have \[\min_{j\in [\deg(v)]}\rho_j(v)\geq c+1.\] 
If, by this point, the $\nonleader$ exit condition $\rho_i(v) - \sigma_i(v) >1$ for $v$ has not been triggered yet, by \Cref{obs:r>=c}, $\counter(v)$ will be incremented to $c+1$, as required. 

Now suppose the $\nonleader$ exit condition $\rho_i(v) - \sigma_i(v) > 1$ is triggered due to some incoming pulse at some point. By \Cref{lem: rho-sigma<=1}, the pulse must be a DFS notification pulse. If this pulse comes from a node in $V \setminus S$, then  (C3) has happened. Otherwise, the pulse comes from a node in $S$, meaning that this node has already exited the synchronized counting phase, so (C2) has happened.
\end{proof}

The following lemma shows that the counting process elects the node  $r$ with the smallest identifier as the leader, as desired.

\begin{lemma}[$r$ is the leader] \label{lem: r exits loop first}
    Let $r$ be the node with the smallest identifier, then $r$ is the first node to exit the synchronized counting phase. Moreover, $r$ sets its state to $\leader$.
\end{lemma}

\begin{proof}
    By \Cref{lem: rho-sigma<=1}, the first node to exit the synchronized counting phase cannot do so via the $\nonleader$ exit condition. Hence it suffices to show that $r$ is the first node that exits the synchronized counting phase.

Apply \Cref{lem: no deadlock} with \( S = V \). Since \( V \setminus S = \emptyset \), (C3) is vacuously false. Moreover, since \( \min_{u \in V \setminus S} \counter(u) = \infty \),  (C1) reduces to \( \min_{v \in S} \counter(v) = \infty \), which cannot hold: By the time the first node \( v \) reaches \( \counter(v) = \id(v) \), it must have already exited the synchronized counting phase and frozen its counter. Therefore, we conclude that (C2) must eventually occur, and hence some node \( v \) eventually exits the synchronized counting phase.

    
    Let \( v \) be the first node to exit the synchronized counting phase. As discussed earlier, it must do so via the \( \leader \) exit condition, which requires \( \counter(v) = \id(v) \).

    Suppose $v \neq r$. Consider the first moment when $\counter(v) = \id(v)$. Consider any path $P$ connecting $v$ and $r$. Since no DFS notification pulse has been delivered so far, by  \Cref{lem: c differ by <=1}, the counter values of any two adjacent nodes in $P$ differ by at most one, so
    \[\id(r) < \id(v) = \counter(v) \leq \counter(r) +  \length(P) < \id(r)+ (N-1),\]
contradicting the fact that all identifiers are integer multiples of $N$, so we must have $v=r$.    
\end{proof}

\paragraph{Remark.} If we simply let the elected leader \( r \) halt after the leader election phase and do not execute \Cref{alg:notify}, we obtain a simple stabilizing leader election algorithm. By \Cref{lem: c differ by <=1}, for any node \( v \in V \setminus \{r\} \), we have \( \counter(v) \leq \counter(r) + \length(P) < \id(r) + (N - 1) < \id(v) \), so the \(\leader\) exit condition cannot be triggered.


\subsection{DFS Notification}\label{subsect:DFS}

We now proceed to analyze the DFS Notification phase. This phase of the algorithm employs a DFS to notify all nodes of the elected leader, while guaranteeing termination with quiescence. Although a graph can have many DFS trees, we use the \emph{specific} tree $T_G$, rooted at the node with the smallest identifier and prioritizing smaller port numbers during traversal, see \Cref{sect:dfs_orientation}. To begin, observe that a DFS visits each edge of the graph exactly once in each direction, as detailed below.


\paragraph{DFS traversal.} The node $r$ with the smallest identifier initiates the DFS traversal, creates a token, and serves as the root of the DFS tree $T_G$. For each node $v$, let $\mathcal{P}(v)$ denote the set of ports from which $v$ has not yet received the token.

When a node $v$ holds the token, it selects the \emph{smallest} port in $\mathcal{P}(v)$ and sends an explore-notification through that port to pass the token. If $v$ is the first node to send the token to a neighbor $u$, then $v$ is designated as the \emph{parent} of $u$ in the DFS tree $T_G$.

If $\mathcal{P}(v) = \emptyset$, it returns the token to its parent by sending a done-notification. In the special case where $v$ is the root, this marks the completion of the entire DFS traversal.

When a node $v$ receives an explore-notification from a neighbor $u$, there are two cases:
\begin{itemize}
    \item If $v$ previously sent the token to $u$ and is now receiving it back, then $\{u, v\}$ is a \emph{tree edge}, and $v$ resumes the DFS traversal as usual.
    \item Otherwise, $\{u, v\}$ is a \emph{back edge}, and $v$ immediately returns the token to $u$ by sending a done-notification.
\end{itemize}

\paragraph{A sequence of directed edges.} We write 
\[A=\left(e_1, e_2, \ldots, e_{2|E|}\right)\]
to denote the sequence of \emph{directed edges} traversed during the DFS. For each edge $e = \{u, v\} \in E$, both directions $(u, v)$ and $(v, u)$ appear exactly once in $A$: The first occurrence corresponds to an explore-notification, and the second to a done-notification.

We now discuss how the DFS traversal is implemented in \Cref{alg:2-connected} and \Cref{alg:notify}.
\begin{description}
    \item[Sender:] The moment when a node $v$ starts sending an explore-notification or a done-notification on port $j$ is marked as $\sendexplore_j(v)$ or $\senddone_j(v)$, respectively. For an explore-notification, node $v$ waits until $\rho_j \geq \leader\id + 1$, and then performs $\sendpulsesuntil_j(\leader\id + N + 2)$. For a done-notification, node $v$ directly performs $\sendpulsesuntil_j(\leader\id + N + 2)$ without any waiting condition.
    
    \item[Receiver:] The moment when a node $v$ confirms the receipt of an explore-notification or a done-notification on port $j$ is marked as $\receiveexplore_j(v)$ and $\receivedone_j(v)$, respectively. The confirmation is triggered when $\rho_j = \leader\id + N + 2$.
\end{description}

If node $v$ is not yet part of the DFS at the time it receives the explore-notification, there is an additional requirement: The pulses arriving on port $j$ must first trigger the $\nonleader$ exit condition. Only after this occurs does $v$ wait until $\rho_j = \leader\id + N + 2$ to confirm receipt of the explore-notification.


The use of $\sendpulsesuntil_j(\leader\id + N + 2)$ for sending both types of notifications serves to ensure quiescence upon termination, as it allows each node to anticipate and recognize the precise number of pulses---$\leader\id + N + 2$---it will receive before becoming inactive.

Assuming that all labeled events in the pseudocode behave as intended, the resulting execution faithfully follows the desired DFS traversal. It therefore remains to verify that the protocol functions correctly: The node with the minimum identifier is elected as the unique leader; every transmission of an explore-notification or done-notification results in confirmation of receipt of that notification; each receipt of an explore-notification or done-notification triggers the subsequent DFS action at that node; and no unintended events occur---for example, a second leader must not be elected, and synchronized counting pulses alone must not erroneously trigger confirmation of receipt for a notification that was never sent.

\paragraph{Notations.} Consider any edge $e = \{u,v\}\in E$. Let $i$ be the port number of $e$ in $u$. Let $j$ be the port number of $e$ in $v$. For ease of presentation, we have the following notations. 
\begin{align*} 
    \mathcal{E}_{u \rightarrow v}^\send &=\begin{cases}
        \sendexplore_i(u) & \text{if $(u,v)$ is before $(v,u)$ in $A$,}\\
        \senddone_i(u) & \text{otherwise.}
    \end{cases}\\
    \mathcal{E}_{u \rightarrow v}^\receive &=\begin{cases}
        \receiveexplore_j(v) & \text{if $(u,v)$ is before $(v,u)$ in $A$,}\\
        \receivedone_j(v) & \text{otherwise.}
    \end{cases}\\
          t_{u \rightarrow v}^\send &= \text{the time when the event $\mathcal{E}_{u \rightarrow v}^\send$ happens.}\\
          t_{u \rightarrow v}^\receive &= \text{the time when the event $\mathcal{E}_{u \rightarrow v}^\receive$ happens.}\\
          t_{r}^\start &= \text{the time when the event $\start(r)$ happens.}
\end{align*}

If an event never occurs, its occurrence time is defined as $\infty$.
We write \[\mathscr{E}^{\star}= \bigcup_{v \in V, \, i \in \deg(v)}\left\{\start(v),\sendexplore_i(v), \receiveexplore_i(v), \senddone_i(v),\receivedone_i(v)\right\}\]
to denote the set of all potential events, including those that are not possible in the DFS tree $T_G$. 

We write $e_j = (u_j, v_j)$ for each $j \in [2|E|]$. Let $R_0 = \{r\}$ and $R_i = \bigcup_{j \in [i]} \{u_j, v_j\}$ for each $i \in [2|E|]$. Intuitively, $R_i \subseteq V$ is the set of nodes that have been visited after traversing the first $i$ directed edges $e_1, e_2, \ldots, e_i$ in $A$.

\paragraph{Induction hypothesis.} The analysis of the DFS notification phase is done by an induction over the indices $i$ in the sequence $A=(e_1, e_2, \ldots, e_{2|E|})$. Intuitively, the induction hypothesis for $i$, denoted as $\IH(i)$, specifies that the DFS traversal up until the traversal of $e_i$ works as expected, with no unintended events occurring. For any $i \in \{0,1,\ldots, 2|E|\}$, $\IH(i)$ is defined as follows. 

\begin{description}
    \item[(H1)] The events associated with the first $i$ directed edges in the DFS traversal happen in order. Specifically, there exists a sequence $t_0 < t_1 < \cdots < t_{i-1} < t_i$ such that
\begin{align*}
t_0 &=t_{r}^\start = t_{u_1 \rightarrow v_1}^\send,\\
t_1 &=  t_{u_1 \rightarrow v_1}^\receive = t_{u_2 \rightarrow v_2}^\send,\\
\vdots\\
t_{i-1} &= t_{u_{i-1} \rightarrow v_{i-1}}^\receive = t_{u_i \rightarrow v_i}^\send,\\
t_i &= t_{u_i \rightarrow v_i}^\receive.
\end{align*}
    \item[(H2)] Before time $t_i$, all events in
    \[\mathscr{E}^\star \setminus \left\{\start(r),\mathcal{E}_{u_1 \rightarrow v_1}^\send,\mathcal{E}_{u_1 \rightarrow v_1}^\receive, \ldots, \mathcal{E}_{u_{i-1} \rightarrow v_{i-1}}^\receive, \mathcal{E}_{u_i \rightarrow v_i}^\send\right\}\]
    do not occur.
    \item[(H3)] At the time $t_i$, all nodes $v \in R_i$ satisfies
    \[\id(r) \leq \counter(v) \leq \id(r)+ \length(P_{v,r}) \leq \id(r)+N-1.\]
    Recall that $P_{v,r}$ is a shortest path from $v$ to $r$ in $\vec{G}$, see \Cref{sect:dfs_orientation}.
\end{description}

For the base case of $i=0$, $\IH(0)$ follows from \Cref{lem: r exits loop first}. Specifically, for $i=0$, (H1) asserts that $t_0 = t_{r}^\start < \infty$,  (H2) asserts that all events in $\mathscr{E}^\star \setminus \left\{t_{r}^\start\right\}$ do not occur before $t_0 = t_{r}^\start$, and (H3) asserts that $\id(r) = \counter(r)$ at the time $t_0 = t_{r}^\start$, since $R_0 = \{r\}$. 

 In the following discussion, we make some observations about the induction hypothesis. For any $i \in \{0,1, \ldots, 2|E|\}$, we write (H1)$(i)$, (H2)$(i)$, and (H3)$(i)$ to refer to the three conditions in $\IH(i)$.

\begin{observation}[Next sending occurs immediately]\label{obs-dfs1}
Let $i \in \{0,1,\ldots,2|E|-1\}$.
If $\IH(i)$ holds, then the following statements hold.
\begin{itemize}
    \item $\mathcal{E}_{u_{i+1} \rightarrow v_{i+1}}^\send$ occurs at time $t_i$.
    \item All events in $\mathscr{E}^\star \setminus \left\{\start(r),\mathcal{E}_{u_1 \rightarrow v_1}^\send,\mathcal{E}_{u_1 \rightarrow v_1}^\receive,  \ldots, \mathcal{E}_{u_i \rightarrow v_i}^\receive, \mathcal{E}_{u_{i+1} \rightarrow v_{i+1}}^\send\right\}$ do not occur at time $t_i$.
\end{itemize}
\end{observation}
\begin{proof}
By (H1)$(i)$ and (H2)$(i)$, the execution of the algorithm follows the desired DFS traversal for the first $i$ directed edges, where the last event $\mathcal{E}$ is $\start(r)$ if $i=0$ and is $\mathcal{E}_{u_i \rightarrow v_i}^\receive$ if $i > 0$. According to the description of our algorithm, immediately after triggering the last event $\mathcal{E}$, the current node $v_i = u_{i+1}$ is ready to continue the DFS traversal by passing the token along the next directed edge $e_{i+1} = (u_{i+1},v_{i+1})$, triggering $t_{u_{i+1} \rightarrow v_{i+1}}^\send$, so the first statement holds.

Other than  $\mathcal{E}$ and $\mathcal{E}_{u_{i+1} \rightarrow v_{i+1}}^\send$, no other event in $\mathscr{E}^\star$ associated with the current node $v_i = u_{i+1}$ is triggered at time $t_i$. The local computation at $v_i = u_{i+1}$ cannot trigger any event in $\mathscr{E}^\star$ associated with other nodes. Hence the second statement holds. 
\end{proof}

\Cref{obs-dfs1} implies that, in the proof of the inductive step below, to establish (H1)\((i+1)\) assuming \( \IH(i) \), we may assume that the event \( \mathcal{E}_{u_{i+1} \rightarrow v_{i+1}}^\send \) has already occurred, and focus on proving that the corresponding event \( \mathcal{E}_{u_{i+1} \rightarrow v_{i+1}}^\receive \) eventually occurs. 


The following observation formalizes the intuition that (H3) guarantees any node visited during the DFS traversal obtains correct knowledge of \( \id(r) \). Consequently, in the subsequent discussion, we may replace \( \leader\id \) with \( \id(r) \).


\begin{observation}[Correct $\leader\id$]\label{obs-dfs2}
Let $i \in \{0,1,\ldots,2|E|\}$.
If $\IH(i)$ holds, then each $v \in R_i$ has already advanced to the DFS notification phase by time $t_i$ with $\leader\id(v) = \id(r)$.
\end{observation}
\begin{proof}
Suppose \( v = r \). By (H1)\((i) \), node \( r \) enters the DFS notification phase via the \( \leader \) exit condition at time \( t_0 \), so \( \leader\id(r) = \id(r) \). 
Suppose \( v \neq r \). By (H1)\((i) \), node \( v \) enters the DFS notification phase via the \( \nonleader \) exit condition at some time before \( t_i \), and sets \( \leader\id(v) \) to \( \left\lfloor \counter(v)/N \right\rfloor \cdot N \). Since the counter value remains unchanged during the DFS notification phase, (H3)\((i) \) implies that \( \leader\id(v) = \id(r) \).
\end{proof}

The following observation follows from the fact that the transmission of all DFS notification pulses from \( u \) to \( v \) must occur entirely within the interval \( \left[ t_{u \rightarrow v}^\send,\, t_{u \rightarrow v}^\receive \right] \). In particular, it implies that, during this interval, node \( u \) must perform the action \( \sendpulsesuntil_{j}(\leader\id + N + 2) \), where \( j \) is the port corresponding to the edge \( e = \{u, v\} \) at node \( u \), and all pulses resulting from this action are delivered to node \( v \).


\begin{observation}[DFS notification pulses]\label{obs-dfs3}
Let $i \in \{0,1,\ldots,2|E|\}$.
If $\IH(i)$ holds, then the following holds for each edge $e=\{u,v\} \in E$ at time $t_i$.
\begin{itemize}
    \item If $(u,v)\in\{e_1, e_2, \ldots, e_i\}$, then exactly $\id(r)+ N+2$ pulses have been delivered from $u$ to $v$, where at least one of them must be a DFS notification pulse.
    \item Otherwise, $u$ has not sent any DFS notification pulse to $v$.
\end{itemize}
\end{observation}
\begin{proof}
If $(u,v)\in\{e_1, e_2, \ldots, e_i\}$, then $t_{u\rightarrow v}^\receive$ has already occurred. From the algorithm description, this implies that $v$ has received all the $\id(r)+ N+2$ pulses sent from $u$. 
If $(u,v)\notin\{e_1, e_2, \ldots, e_i\}$, then (H2)$(i)$ implies that $u$ has not sent any DFS notification pulse to $v$.

We discuss one subtle issue in the proof: While the upper bound $\id(r)+ N+2$ appears to be clear, there might be a possibility that $u$ already sends to $v$ more than this number of pulses while $u$ is in the synchronized counting phase. This is however impossible because the number of such pulses is $\counter(u)+1 \leq (\id(r)+N-1)+1 < \id(r)+ N+2$ according to \Cref{obs:s=c+1} and (H3)$(i)$. This also implies that some of the $\id(r)+ N+2$ pulses sent from $u$ to $v$ must be DFS notification pulses.
\end{proof}

We now state the lemma for the inductive step.

\begin{lemma}\label{lem:inductive-step}
Let $i \in \{0,1,\ldots,2|E|-1\}$.
If $\IH(i)$ holds, then $\IH(i+1)$ holds.
\end{lemma}

Before proving \Cref{lem:inductive-step}, we first use it to prove the main theorem.

\begin{proof}[Proof of \Cref{thm:main}]
The base case $\IH(0)$ follows from \Cref{lem: r exits loop first}, which shows that $\start(r)$ occurs after finite time and is before all other events in $\mathscr{E}^\star$, satisfying (H1) and (H2). The condition (H3) is satisfied because $r$  enters the DFS notification phase through the $\leader$ exit condition, 
which ensures $\counter(r)=\id(r)=\id(r)+P_{r,r}$, as $P_{r,r} = 0$.

By \Cref{lem:inductive-step}, $\IH(1), \IH(2), \ldots, \IH(2|E|)$ hold. 
Given that $\IH(2|E|)$ holds, we show that the algorithm terminates by time $t_{2|E|}$ with quiescence correctly.
By \Cref{obs-dfs2}, all nodes $v$ have advanced to the DFS notification phase with $\leader\id(v) = \id(r) \neq \id(v)$. This also implies that all of the nodes in $V \setminus \{r\}$ are non-leaders, so $r$ is the unique leader.

Given that $\IH(2|E|)$ holds, each edge is traversed exactly once in both directions during the DFS. For each node $v$, $\mathcal{E}_{u\rightarrow v}^\receive$ has been triggered for each of its neighbors $u$, meaning that $\mathcal{P}(v)=\emptyset$ and $v$ has terminated by the time $t_{2|E|}$. Moreover, the termination comes with a quiescence guarantee by \Cref{obs-dfs2}.

The message complexity bound follows from \Cref{obs-dfs2}, which states that the number of pulses sent along each direction of each edge is exactly $\id(r)+N+2$, where $\id(r) = \idmin \cdot N$, as we multiply each node identifier by $N$ at the beginning. Therefore, the overall message complexity is $O(m \cdot N \cdot \idmin)$.  

The fact that the leader \( r \) is the final node to terminate follows from the ordering of events specified in (H1)\((2|E|)\).
\end{proof}

\paragraph{Inductive step.} For the rest of the discussion, we prove \Cref{lem:inductive-step}. Suppose $\IH(i)$ holds. Our goal is to show that $\IH(i+1)$ holds. By \Cref{obs-dfs1}, we may assume that
\[t_i = t_{u_{i+1} \rightarrow v_{i+1}}^\send < \infty\]
and take it as the starting point of the analysis. We want to show that $\mathcal{E}_{u_{i+1} \rightarrow v_{i+1}}^\receive$ happens eventually, before triggering any other events in $\mathscr{E}^\star$. The proof consists of two parts. In the first part, we show that, before $\mathcal{E}_{u_{i+1} \rightarrow v_{i+1}}^\receive$, no other event in $\mathscr{E}^\star$ can happen, so (H2)$(i+1)$ holds. We also show that the following \emph{small counter property} holds for not only the nodes in the DFS notification phase but also the remaining nodes in the synchronized counting phase. In the second part, we show that $\mathcal{E}_{u_{i+1} \rightarrow v_{i+1}}^\receive$ occurs within finite time, so (H1)$(i+1)$ holds. Moreover, if this is also the moment where $v_{i+1}$ transitions from the synchronized counting phase to the DFS notification phase, then the counter value of $v_{i+1}$ satisfies H3$(i+1)$.

\begin{definition}[Small counter property]
A node $v$ is said to have the small counter property if \[\counter(v) \leq \id(r)+\length(P_{v,r}) < \id(v).\]    
\end{definition}

\paragraph{Part 1: Bad events do not happen.}
We consider two bad events.
\begin{description}
    \item[$\mathcal{E}_A$:] The small counter property is violated for some node $v \in V$.
    \item[$\mathcal{E}_B$:] Some event in $\mathscr{E}^\star \setminus \left\{\start(r),\mathcal{E}_{u_1 \rightarrow v_1}^\send,\mathcal{E}_{u_1 \rightarrow v_1}^\receive,  \ldots, \mathcal{E}_{u_i \rightarrow v_i}^\receive, \mathcal{E}_{u_{i+1} \rightarrow v_{i+1}}^\send\right\}$ is triggered.
\end{description}
We aim to show that both bad events cannot happen within the time interval $\left[t_i, t_{u_{i+1} \rightarrow v_{i+1}}^\receive\right)$, and moreover $\mathcal{E}_A$ does not happen at time $t_{u_{i+1} \rightarrow v_{i+1}}^\receive$.

\begin{lemma}[Quiescence]\label{lem-bad-1}
Suppose $\IH(i)$ holds.
    Consider a moment $t \geq t_i$. Suppose  $\mathcal{E}_B$ has not happened, then at any moment within the time interval $[t_i, t]$, there is no DFS notification pulse in transit, except for the directed edge $e_{i+1} = (u_{i+1}, v_{i+1})$.
\end{lemma}
\begin{proof}
Suppose there is such a DFS notification pulse. Observe that sending such a pulse requires first triggering some event in $\mathscr{E}^\star \setminus \left\{\start(r),\mathcal{E}_{u_1 \rightarrow v_1}^\send,\mathcal{E}_{u_1 \rightarrow v_1}^\receive,  \ldots, \mathcal{E}_{u_i \rightarrow v_i}^\receive, \mathcal{E}_{u_{i+1} \rightarrow v_{i+1}}^\send\right\}$, contradicting our assumption that $\mathcal{E}_B$ has not happened by time $t$.  
\end{proof}

\begin{lemma}[Chain and anchor]\label{lem-bad-2}
Suppose $\IH(i)$ holds.
    Consider a moment $t \geq t_i$. Suppose no DFS notification pulse has been delivered to any node in $V \setminus R_i$, except that there might be DFS notification pulse delivered along the directed edge $e_{i+1} = (u_{i+1}, v_{i+1})$, then all nodes $v \in V$ satisfy the small counter property.
\end{lemma}
\begin{proof}
If $v \in R_i$, then the small counter property follows from (H3)$(i)$.
Now consider any $v \in V \setminus R_i$.
We use the chain-and-anchor argument. Consider the first node $w$ in the shortest path $P_{v,r}$ that belongs to $R_i$, so \[\length(P_{v,r})=\length(P_{v,w})+\length(P_{w,r}).\]
By (H3)$(i)$, $\counter(w) \leq \id(r)+\length(P_{w,r})$.
By \Cref{lem: c differ by <=1}, $\counter(v)$ is at most
\[\counter(w) + \length(P_{v,w}) \leq \id(r)+\length(P_{w,r})+ \length(P_{v,w}) = \id(r)+\length(P_{v,r}),\]
as required.

For the calculation to work, we need to ensure that no DFS notification pulse has been delivered along any edge in $P_{v,w}$ in any direction. This requires that $P_{v,w}$ does not traverse along the edge $\{u_{i+1}, v_{i+1}\}$. As $u_{i+1} \in R_i$, such a traversal is possible only when $v_{i+1} \notin R_i$, in which case the edge is oriented as $(u_{i+1}, v_{i+1})$ in $\vec{G}$, so it is impossible that $P_{v,w}$ contains $(u_{i+1}, v_{i+1})$, since otherwise we would have selected $w = u_{i+1}$.
\end{proof}

\begin{lemma}\label{lem-bad-3}
    Suppose $\IH(i)$ holds.
    Consider a moment $t \geq t_i$.  If both $\mathcal{E}_A$ and $\mathcal{E}_B$ have not happened, then the following statements hold.
  \begin{itemize}
      \item The next pulse delivery also does not trigger $\mathcal{E}_A$.
      \item If the next pulse delivery does not trigger $\mathcal{E}_{u_{i+1} \rightarrow v_{i+1}}^\receive$, then it also does not trigger $\mathcal{E}_B$.
  \end{itemize}  
    \end{lemma}
\begin{proof}
By \Cref{lem-bad-1}, the next pulse delivered---as well as any pulse delivered after time $t_i$ and before time $t$---is either a DFS notification pulse along the directed edge $e_{i+1}=(u_{i+1}, v_{i+1})$ or a synchronized counting pulse. Combining this with \Cref{obs-dfs3}, the precondition of \Cref{lem-bad-2} is satisfied at the moment after the next pulse delivery, so \Cref{lem-bad-2} implies that the next pulse delivery does not trigger $\mathcal{E}_A$.

Suppose the next pulse delivery triggers an event $\mathcal{E}$ in \[\mathscr{E}^\star \setminus \left\{\start(r),\mathcal{E}_{u_1 \rightarrow v_1}^\send,\mathcal{E}_{u_1 \rightarrow v_1}^\receive,  \ldots, \mathcal{E}_{u_i \rightarrow v_i}^\receive, \mathcal{E}_{u_{i+1} \rightarrow v_{i+1}}^\send\right\}.\]
Let $v$ be the node that receives this pulse.

Since the pulse delivery does not trigger $\mathcal{E}_A$, we know that the counter value at any node $v$ after the pulse delivery is at most $\id(r)+N-1$, so $\mathcal{E}$ cannot be $\start(v)$, as $\id(r)+N-1 < \id(v)$ for any $v \in V \setminus \{r\}$.

If the pulse is a DFS notification pulse along the directed edge $e_{i+1}=(u_{i+1}, v_{i+1})$, then $\mathcal{E}$ must be $\mathcal{E}_{u_{i+1} \rightarrow v_{i+1}}^\receive$. Therefore, for the rest of the discussion, we assume that the pulse is a synchronized counting pulse. In particular, \Cref{lem: rho-sigma<=1} implies that the pulse cannot trigger the $\nonleader$ exit condition.

Therefore, the only remaining way to trigger any event in $\mathscr{E}^\star$ is to get a total of $\id(r)+N+2$ pulses via some port, but that is not possible with synchronized counting pulses only, unless the small counter property is violated. Specifically, \Cref{obs:s=c+1} implies that the number of synchronized counting pulses sent along a port equals the counter value plus one. While the small counter property holds, this number if capped at $(\id(r)+N-1) +1$, which is smaller than the required $\id(r)+N+2$. 

Therefore, we conclude that, if the next pulse delivery does not trigger $\mathcal{E}_{u_{i+1} \rightarrow v_{i+1}}^\receive$, then it does not trigger $\mathcal{E}_B$.
\end{proof}

\begin{lemma}[Bad events do not happen]\label{lem:bad_avoid}
Event $\mathcal{E}_A$ does not happen within the time interval $\left[t_i, t_{u_{i+1} \rightarrow v_{i+1}}^\receive\right]$, and event $\mathcal{E}_B$ does not happen within the time interval $\left[t_i, t_{u_{i+1} \rightarrow v_{i+1}}^\receive\right)$.
\end{lemma}
\begin{proof}
    We prove the lemma by an induction over the number of pulses sent after time $t_i$. 
    
    For the base case of the induction, consider the moment $t = t_i$. By \Cref{obs-dfs1}, $\mathcal{E}_B$ has not happened yet. By a combination of \Cref{obs-dfs3} and \Cref{lem-bad-1}, the precondition of \Cref{lem-bad-2} is satisfied at $t = t_i$, so $\mathcal{E}_A$ has not happened yet.

    For the inductive step, apply \Cref{lem-bad-3} to any moment $t \in \left[t_i, t_{u_{i+1} \rightarrow v_{i+1}}^\receive\right)$. Suppose, by induction hypothesis that both $\mathcal{E}_A$ and $\mathcal{E}_B$ have not happened by time $t$, then it is guaranteed that the next pulse delivery does not trigger $\mathcal{E}_A$. Moreover, if the next pulse delivery does not trigger $\mathcal{E}_{u_{i+1} \rightarrow v_{i+1}}^\receive$, then it also does not trigger  $\mathcal{E}_B$.
     Therefore, $\mathcal{E}_A$ does not occur in the time interval $\left[t_i, t_{u_{i+1} \rightarrow v_{i+1}}^\receive\right]$, and  $\mathcal{E}_B$ does not occur in the time interval $\left[t_i, t_{u_{i+1} \rightarrow v_{i+1}}^\receive\right)$.
\end{proof}

\paragraph{Part 2: The next event happens.} After getting rid of the bad events $\mathcal{E}_A$ and  $\mathcal{E}_B$, the next thing to do is to show that $\mathcal{E}_{u_{i+1} \rightarrow v_{i+1}}^\receive$ happens within finite time. 

\begin{lemma}[The next event happens]\label{lem:good}
If $\IH(i)$ holds, then $\mathcal{E}_{u_{i+1} \rightarrow v_{i+1}}^\receive$ occurs within finite time. Moreover, $\counter(v_{i+1}) \geq \id(r)$ at time $t_{u_{i+1} \rightarrow v_{i+1}}^\receive$.
\end{lemma}
\begin{proof}
Let the port of the edge $\{u_{i+1}, v_{i+1}\}$ be $j$ in $u_{i+1}$ and $k$ in $v_{i+1}$. By \Cref{obs-dfs1}, $\mathcal{E}_{u_{i+1} \rightarrow v_{i+1}}^\send$ occurs at time $t_i$.
We break the analysis into a few cases.

\paragraph{Done-notification.} Suppose a done-notification is sent across $e_{i+1}$, meaning that $\mathcal{E}_{u_{i+1} \rightarrow v_{i+1}}^\send$ is $\senddone_j(v_{i+1})$ and $\mathcal{E}_{u_{i+1} \rightarrow v_{i+1}}^\receive$ is $\receivedone_k(v_{i+1})$.
The way our algorithm sends a done-notification implies that a total of $\id(r)+N+2$ pulses have been sent from $u_{i+1}$ along the edge $\{u_{i+1}, v_{i+1}\}$ by time $t_i$. To confirm the receipt of the done-notification, it suffices that all these $\id(r)+N+2$ pulses are delivered $v_{i+1}$, which happens within a finite time.

\paragraph{Explore-notification.}  Suppose an explore-notification is sent across $e_{i+1}$. That is, 
 $\mathcal{E}_{u_{i+1} \rightarrow v_{i+1}}^\send$ is $\sendexplore_j(v_{i+1})$ and $\mathcal{E}_{u_{i+1} \rightarrow v_{i+1}}^\receive$ is $\receiveexplore_k(v_{i+1})$.
By the description of our algorithm, upon triggering the event $\mathcal{E}_{u_{i+1} \rightarrow v_{i+1}}^\send$, $u_{i+1}$ first waits until getting at least $\id(r)+1$ messages from $v_{i+1}$. After that, $u_{i+1}$ sends to $v_{i+1}$ several pulses until a total of $\id(r)+N+2$ pulses have been sent along the directed edge $e_{i+1}=(u_{i+1},v_{i+1})$. 

\paragraph{Back edge.} Suppose $\{u_{i+1}, v_{i+1}\}$ is a back edge, meaning that $v_{i+1}$ has already advanced to the DFS notification phase by time $t_i$. By (H3)$(i)$, $\counter(v_{i+1}) \geq \id(r)$, so $v_{i+1}$ has already sent to $u_{i+1}$ at least $\counter(v_{i+1})+1 \geq \id(r)+1$ synchronized counting pulses (\Cref{obs:s=c+1}), meaning that $u_{i+1}$ only waits for a finite amount of time, and eventually $v_{i+1}$ will receive a total of $\id(r)+N+2$ pulses from $u_{i+1}$, triggering $\mathcal{E}_{u_{i+1} \rightarrow v_{i+1}}^\receive$.

\paragraph{Tree edge.} Suppose $\{u_{i+1}, v_{i+1}\}$ is a tree edge, meaning that $v_{i+1}$ is still in the synchronized counting phase by time $t_i$. 

First, we claim that the waiting before sending pulses takes finite time for $u_{i+1}$ by applying \Cref{lem: no deadlock} to the moment $t_i$ with $S = V \setminus R_i$. At least one of the three types of progress is made within finite time.
\begin{itemize}
    \item Suppose (C1) holds, then 
\[1+\id(r)\leq 1 + \min_{u \in V \setminus S} \counter(u) \leq \min_{v \in S} \counter(v) \leq \counter(v_{i+1}),\]
where the first inequality follows from (H3)$(i)$. By \Cref{obs:s=c+1}, this means that $v_{i+1}$ has sent out $\id(r)+1$ pulses to $u_{i+1}$, and eventually these pulses will arrive at $u_{i+1}$, so $u_{i+1}$ only waits for a finite time.
\item Suppose (C2) holds, meaning that some node in $S$ exits the synchronized counting phase, so either $\mathcal{E}_{u_{i+1} \rightarrow v_{i+1}}^\receive$ is triggered, or some event in \[\mathscr{E}^\star \setminus \left\{\start(r),\mathcal{E}_{u_1 \rightarrow v_1}^\send,\mathcal{E}_{u_1 \rightarrow v_1}^\receive,  \ldots, \mathcal{E}_{u_i \rightarrow v_i}^\receive, \mathcal{E}_{u_{i+1} \rightarrow v_{i+1}}^\send\right\}\] is triggered. In the former case, we are done already. The latter case is not possible before first triggering $\mathcal{E}_{u_{i+1} \rightarrow v_{i+1}}^\receive$ in view of \Cref{lem:bad_avoid}.
\item Suppose (C3) holds, meaning that a DFS notification pulse from some node in $V \setminus S$ is delivered to some node in $S$. By a combination of \Cref{lem-bad-1} and \Cref{lem:bad_avoid}, we infer that unless $\mathcal{E}_{u_{i+1} \rightarrow v_{i+1}}^\receive$ has been triggered, the pulse must be sent from $u_{i+1}$ to $v_{i+1}$. In either case, the waiting must be done within finite time.
\end{itemize}

After the waiting is done,  eventually $v_{i+1}$ will receive a total of $\id(r)+N+2$ pulses from $u_{i+1}$, meaning that
\[\rho_k(v_{i+1}) = \id(r)+N+2.\]
For $v_{i+1}$ to confirm the receipt of the explore-notification, we need $v_{i+1}$ to first satisfy the $\nonleader$ exit condition, To do so, we use the small counter property, which is guaranteed to hold during the time interval $\left[t_i, t_{u_{i+1} \rightarrow v_{i+1}}^\receive\right)$ in view of \Cref{lem:bad_avoid}. 

Due to the small counter property, the counter value of $v_{i+1}$ is capped by $\id(r)+N-1$, so by \Cref{obs:s=c+1},  while $v_{i+1}$ is still in the synchronized counting phase,
\[\sigma_k(v_{i+1}) \leq (\id(r)+N-1)+1=\id(r)+N.\]
Therefore, there must be a moment when the $\nonleader$ exit condition $\rho_k - \sigma_k > 1$ is satisfied for $v_{i+1}$, and then once all the $\id(r)+N+2$ pulses from $u_{i+1}$ have been delivered to $v_{i+1}$, $\mathcal{E}_{u_{i+1} \rightarrow v_{i+1}}^\receive = \receiveexplore_k(v_{i+1})$ is triggered.
\end{proof}


\paragraph{Part 3: Wrapping up.} Now we have collected all the needed ingredients to prove \Cref{lem:inductive-step}.

\begin{proof}[Proof of \Cref{lem:inductive-step}.]
Let $i \in \{0,1,\ldots,2|E|-1\}$.
Suppose $\IH(i)$ holds. Our goal is to show that $\IH(i+1)$ holds. To prove (H1)$(i+1)$, it suffices to show that
\[t_{i} = t_{u_{i+1} \rightarrow v_{i+1}}^\send < t_{u_{i+1} \rightarrow v_{i+1}}^\receive < \infty.\]
The part $t_{i} = t_{u_{i+1} \rightarrow v_{i+1}}^\send < t_{u_{i+1} \rightarrow v_{i+1}}^\receive$  follows from \Cref{obs-dfs1}. The part
 $t_{u_{i+1} \rightarrow v_{i+1}}^\receive < \infty$ follows from \Cref{lem:good}. We set $t_{i+1} = t_{u_{i+1} \rightarrow v_{i+1}}^\receive$ To prove (H2)$(i+1)$, it suffices to show that, within the timer interval $[t_i, t_{i+1})$, all events in  $\mathcal{E}_B$ do not occur. This follows from \Cref{lem:bad_avoid}. To prove (H3)$(i+1)$, it suffices to show that \[\id(r) \leq \counter(v_{i+1}) \leq \id(r)+ \length(P_{v_{i+1},r})\]
 at time $t_{i+1}$.
 If we already have $v_{i+1} \in R_i$, then this follows from (H3)$(i)$. Otherwise, by \Cref{lem:bad_avoid} again, we infer that the small counter property holds for $v_{i+1}$ at time $t_{i+1}$, so $\counter(v_{i+1}) \leq \id(r)+ \length(P_{v_{i+1},r})$. The inequality $\id(r) \leq \counter(v_{i+1})$ follows from \Cref{lem:good}.
\end{proof}

\section{Uniform Leader  Election for Unoriented Rings}\label{sec:unoriented}

In this section, we prove our main theorem on rings.
\secondmainthm*

The proof is given by our uniform algorithm, whose pseudocode is given in
Algorithm~\ref{algo:unoriented}. Differently, from the algorithms of Section \ref{sec:qle2ed} in this section we do not use an event driven formalism. This choice is made to simplify the presentation of the algorithm, which proceeds in several phases that are more easily explained by explicitly using blocking primitives that wait until a pulse on a specific port is received.

Therefore, we introduce a new action, {\boldmath $\rcvpulse_i$}: Wait until a pulse is received from port $i$, where $i \in \{0,1,q\}$.  
    When $i = q$, the procedure waits for the first pulse on any port and stores the corresponding port number in the variable $q$.  

In addition to waiting for a pulse, we use:
{\boldmath $\sendpulses_i(k)$}, send $k$ pulses through port $i$, where $i \in \{0,1\}$.

We also define the local variable {\boldmath $\diff(v)$}, which stores the difference between the number of pulses node $v$ received on port $1$ and the number of pulses received on port $0$.

\paragraph{Algorithm description.} Our algorithm is decomposed in different phases that we describe below.

\begin{itemize}
\item The \emph{initialization}. First, each node doubles its
  identifier $\id$ (Line~\ref{line-algo-unoriented-init}). The
  purpose of this action is to ensure that no two nodes have
  consecutive identifiers. This allows us to distinguish the different
  phases of the algorithm during the execution.

\item The \emph{competing} phase (Lines~\ref{line-algo-unoriented-competing-beg} to~\ref{line-algo-unoriented-competing-end} in blue). Every node
  starts by exchanging pulses with both its neighbors in a
  pseudo-synchronous way (sending them pulses and waiting for their
  responses). The number of times a node repeats this step is equal
  to the value of its identifier $\id$.

\item The \emph{solitude-checking} phase (Lines~\ref{line-algo-unoriented-check-alone-beg} to~\ref{line-algo-unoriented-check-alone-end} in orange).
  Once a node has finished the competing phase, it wants to know if
  it was the last one competing. To do so, it sends two pulses on
  port $0$, receives one on port $1$ and waits for another pulse. If this
  last pulse arrives on port $1$, it means that its own pulse has
  made a complete traversal of the cycle and that all other nodes
  stopped competing. Otherwise (i.e., if it gets a pulse on port $0$
  before getting two pulses on port $1$), it means that there is
  still another node executing the competing phase.

\item The \emph{global termination} phase
  (Lines~\ref{line-algo-unoriented-termination-beg}
  to~\ref{line-algo-unoriented-termination-end} in green). If a node
  detects that all other nodes stopped competing (i.e., if $q = 1$
  on Line~\ref{line-algo-unoriented-test-alone}), then it knows it
  will be the leader. Before terminating, it sends a pulse on port
  $0$ to inform the other nodes that the execution is over. This
  termination pulse will traverse the entire ring and come back through port
  $1$. At this point, every other node will eventually be in state
  $\nonleader$ and there will be no pulses in transit. Then the
  node can enter the state $\leader$ and terminate. 

\begin{algorithm}[ht!]
\footnotesize
  \DontPrintSemicolon
  \SetKw{Send}{send}
  \SetKw{Receive}{receive}
  \SetKwData{Status}{status}
  \SetKwFunction{Leader}{Leader}
  \SetKwFunction{NonLeader}{Non-Leader}
  \SetKwData{Count}{count}
  \SetKwData{ID}{ID}
  $\ID(v) \leftarrow 2*\ID(v)$ \label{line-algo-unoriented-init}\;
  \color{blue}
  \For{$i \leftarrow 1$ \KwTo $\ID(v)$
    \label{line-algo-unoriented-competing-beg}}     
    {
      $\sendall(1)$ \tcp*{send 1 pulse on both ports} \label{line-algo-unoriented-competing-send0}
      $\rcvpulse_0$ \tcp*{wait for a pulse from port 0} \label{line-algo-unoriented-competing-rec0}
      $\rcvpulse_1$ \tcp*{wait for a pulse from port 1} \label{line-algo-unoriented-competing-rec1} \label{line-algo-unoriented-competing-end}
    }
  \color{black}
  \tcc{Check if I am alone, i.e., if I have the largest identifier}
  \color{orange}
  $\sendpulses_0(2)$ \tcp*{send 2 pulses on port 0} \label{line-algo-unoriented-check-alone-beg}
  $\rcvpulse_1$ \tcp*{wait for a pulse from port 1} \label{line-algo-unoriented-check-alone-rec1}
  $\rcvpulse_q$ \tcp*{wait for first pulse from any port; store port in $q$} \label{line-algo-unoriented-check-alone-end}
  \color{black}
  \eIf{$q = 1$ \label{line-algo-unoriented-test-alone}}
    {
    \tcc{I am the last one competing}
    \color{teal}
    $\sendpulses_0(1)$ \tcp*{send termination pulse on port 0 — I’m leader} \label{line-algo-unoriented-termination-beg}
    $\rcvpulse_1$ \tcp*{wait for termination pulse to return from port 1} \label{line-algo-unoriented-termination-end}
    \color{black}
    \Return{\Leader} \label{line-algo-unoriented-return-leader}
    \color{black}
  }
  {
    \tcc{There is someone with a larger identifier}
    \color{red}
    \color{violet} 
     $\sendpulses_1(2)$ \tcp*{send 2 pulses on port 1 — balance line 7} \label{line-algo-unoriented-rebalancing-beg}
    $\rcvpulse_0$ \tcp*{wait for a pulse from port 0} \label{line-algo-unoriented-rebalancing-rec0}
    $\rcvpulse_1$ \tcp*{wait for a pulse from port 1} \label{line-algo-unoriented-rebalancing-end}
    \color{red} 
    $\diff \leftarrow 0$  \label{line-algo-unoriented-relaying-beg}\; 
    \color{black}
    \tcc{$\diff$ = pulses from port 1 minus pulses from port 0}
       \color{red}   
    \Repeat
      {$|\diff| \geq 3$ \label{line-algo-unoriented-relaying-end}}
      {\label{line-algo-unoriented-relaying-loop-beg}
        $\rcvpulse_q$ \tcp*{wait for pulse from any port; store port in $q$}
        $\diff \leftarrow \diff + 2*q - 1$ \tcp*{update difference counter}
        $\sendpulses_{1-q}(1)$ \tcp*{forward pulse to opposite port}
      }
    \color{black}
    \Return{\NonLeader} \label{line-algo-unoriented-return-nonleader}
    \color{black}
  }
  \caption{Uniform leader election on unoriented rings for node $v$}
  \label{algo:unoriented}
\end{algorithm}

\item The \emph{rebalancing} phase
  (Lines~\ref{line-algo-unoriented-rebalancing-beg}
  to~\ref{line-algo-unoriented-rebalancing-end} in violet). If the
  node is not alone (i.e., if $q = 0$ on
  Line~\ref{line-algo-unoriented-test-alone}), we want the node to
  relay pulses between its competing neighbors so that these
  neighbors act as if there was no node between them. To do so, we
  have to ensure that it sent the same number of pulses in both
  directions. So the node sends two pulses on port 1 (to balance
  the pulses sent on port $0$ on
  Line~\ref{line-algo-unoriented-check-alone-beg}). At this point, on
  each port, the node has sent one more pulse than it has
  received. So it waits for one pulse on each port before entering
  the relaying phase. For its neighbors, the node behaves as if it
  has executed two more iterations of the competing phase: it has sent
  two pulses on each port and has received two pulses on each
  port. Since no two nodes have consecutive identifiers (recall
  that each node doubled its $\id$ at
  Line~\ref{line-algo-unoriented-init}), no other node can
  terminate its competition phase before receiving both these
  pulses.
  
\item The \emph{relaying} phase
  (Lines~\ref{line-algo-unoriented-relaying-beg}
  to~\ref{line-algo-unoriented-relaying-end} in red). In this
  phase, the node just relays the pulses it gets on port $0$ to
  port $1$ and vice-versa. It also keeps a counter ($\diff$) to store the
  difference between the numbers of pulses received on port $0$ and
  on port $1$. When this counter reaches $3$ or $-3$, then it means
  that one node has successfully passed the solitude test and sent
  a termination pulse. The node can then enter the state
  $\nonleader$ and terminate.
\end{itemize}

\subsection{Correctness}

  We say that a node is \emph{active} if it is executing an instruction
  between Lines~\ref{line-algo-unoriented-competing-beg}
  and~\ref{line-algo-unoriented-rebalancing-end}. A node is a
  \emph{relaying node} if is executing an instruction between
  Lines~\ref{line-algo-unoriented-relaying-beg}
  and~\ref{line-algo-unoriented-relaying-end}. 

Consider two rings: $R$ of size $n$ and $R'$ of size $n-1$, where $R'$ is obtained by removing a node $v$ from $R$ and directly connecting its neighbours $w$ and $u$. We say that an execution of our algorithm on $R$ is indistinguishable from an execution on $R'$ if there exists an execution on $R'$ in which $w$ and $u$ return from waiting on the $\rcvpulse$ actions at exactly the same times as they do in $R$.

Notice that, in such executions, $w$ and $u$ do not perceive any difference between the execution of the algorithm on $R$ and on $R'$. This, in turn, implies that all other nodes also observe an execution, when viewed in terms of $\rcvpulse$ events, that could occur in either $R$ or $R'$. 
\color{black}

\begin{proposition}\label{prop:redux}
Consider a ring of size $n \geq 2$ with consecutive nodes $v_1, v_2$ and $v_3$ (with $v_1=v_3$ if $n=2$)  such that $v_2$ is the node with the smallest identifier on the entire ring. As long as $v_2$ does not stop, the execution of the algorithm on this ring is indistinguishable for any node on the ring from a reduced ring of size $n-1$, where $v_2$ has been removed and $v_1$ is directly connected with $v_3$. 
\end{proposition}

Without loss of generality, let us assume that at $v_2$, the port
number to $v_1$ is $0$ and the port number to $v_3$ is $1$. Let
$\id_1$, $\id_2$, and $\id_3$ be the respective identifiers of
nodes $v_1, v_2, v_3$. To prove Proposition~\ref{prop:redux}, we
first establish several observations and lemmas.

\begin{observation}\label{observation:obvious1}
    When a node ends the $i$-th iteration of the loop of the competing phase, it has sent $i$ pulses to each of its neighbors and has received $i$ pulses from each of them.
\end{observation}

\begin{lemma}\label{lemma:sync}
  Consider two neighboring nodes, $v_a$ and $v_b$, executing the
  competing phase. Given their respective indices, $i_a$ and $i_b$, in
  the  loop at Lines~\ref{line-algo-unoriented-competing-beg}--\ref{line-algo-unoriented-competing-end}, we have
  $|i_a - i_b| \leq 1$. Moreover, if $i_b \leq i_a$ and $v_b$ has not
  executed the sends of its iteration $i_b$ then $i_b=i_a$.
\end{lemma}
\begin{proof}

Suppose w.l.o.g. that node $v_a$ exchanges pulses with $v_b$ on port 1, while node $v_b$ exchanges pulses with $v_a$ on port 0.

If node $v_a$ executed $x$ times
Line~\ref{line-algo-unoriented-competing-rec1} (i.e.,
$x+1 \geq i_a \geq x$), then node $v_b$ executed $y \geq x$ times
Line~\ref{line-algo-unoriented-competing-send0} (i.e., $i_b \geq
y$). From this $i_b \geq i_a-1$.  By a symmetric argument we have
$i_a \geq i_b-1$ and this shows the first part of the lemma.  Assume
now that $v_b$ has not executed its send, this means that $y=(i_b-1)$
and thus $x \leq y \leq i_b-1$ therefore since
$i_a \leq x+1 \leq i_b$. For the hypothesis we have $i_a=i_b$.
\end{proof}

\begin{lemma}\label{lemma:minexit}
Among nodes $v_1,v_2$ and $v_3$, node $v_2$ is the first  to exit the competing
  phase. Moreover, $v_1$ (resp. $v_3$) cannot exit its competing phase
before it has received two pulses sent by $v_2$ after
 $v_2$ has exited the competing phase (sent at
  Line~\ref{line-algo-unoriented-check-alone-beg} for $v_1$ and
  at~\ref{line-algo-unoriented-rebalancing-beg} for $v_3$).
\end{lemma}

\begin{proof}
By Lemma \ref{lemma:sync} and Observation \ref{observation:obvious1}, we have that when $v_2$ exits the loop, it has sent  $\id_{2}$ pulses to both its neighbors and thus $v_1$ and $v_3$ are in iteration $\id_{2}$ or $\id_{2}+1$ of the loop. Since their identifiers are at least $\id_{2}+2$, they are still in their competing phase. 
Note that, by Lemma \ref{lemma:sync} and Observation \ref{observation:obvious1}, they cannot complete iteration $\id_{2}+2$ before they have received at least $\id_{2}+2$ pulses from $v_2$. The last two of these pulses have to be sent by $v_2$ either at Line~\ref{line-algo-unoriented-check-alone-beg} or at Line~\ref{line-algo-unoriented-rebalancing-beg}.
\end{proof}

\begin{lemma}\label{lemma:nomaxnoleader}
  Node $v_2$ eventually executes
  Line~\ref{line-algo-unoriented-test-alone}, i.e., it receives a
  pulse on port $1$ and then another pulse on some port
  $q$. Moreover, the second pulse delivered arrives on port $q = 0$
  and $v_2$ enters the rebalancing phase.
\end{lemma}

\begin{proof}
  Any node $v \neq v_2$ in the ring has an $\id$ that is at least
  $\id_{2}+2$. Consequently, by Lemma \ref{lemma:minexit} and Lemma \ref{lemma:sync}, every node eventually completes $\id_2$
  iterations of the competing phase and every node $v \neq v_2$ enters iteration $\id_2+1$. Thus, in particular, $v_1$ and $v_3$ send pulses to $v_2$ in iteration $\id_{2}+1$. Consequently,  $v_2$ will eventually receive a pulse from $v_3$ on port $1$ (i.e., it will not be blocked at Line~\ref{line-algo-unoriented-check-alone-rec1}). 
  Moreover, since $v_1$ has sent a pulse \msg to $v_2$ at iteration $\id_{2}+1$, $v_2$ will eventually receive a pulse on some port $q$ (i.e., it will execute Line~\ref{line-algo-unoriented-check-alone-end}). As long as $v_2$ does not send a pulse to $v_3$, $v_3$ will not complete iteration $\id_2+1$ of the competing phase, and it will thus not send another pulse to $v_2$ before \msg is delivered. Consequently, $q = 0$ and $v_2$ enters the rebalancing phase. 
\end{proof}

Notice that if $v_2$ is only a local minimum among $v_1, v_2,$ and $v_3$, rather than a global one, it is still possible to adapt the proof to obtain a weaker variant of the lemma. In this variant, if $v_2$ executes Line~\ref{line-algo-unoriented-test-alone}, it must do so with $q = 0$.

From Lemma \ref{lemma:nomaxnoleader} we have:
\begin{corollary}\label{cor:loser}
  Node $v_2$ does not enter the global termination phase and thus it cannot terminate the algorithm in the $\leader$ state. 
\end{corollary}

The following observation follows from the fact that the second minimum identifier in the ring is at least $\id_2+2$. Recall that the reduced ring is obtained by removing node $v_2$ and directly connecting $v_1$ and $v_3$ while respecting the original port labeling.

\begin{observation}\label{obs:reducing}
When the algorithm is executed on the reduced ring, $v_1$ and $v_3$ send at least $\id_2+2$ pulses to each other (and they receive these pulses) during their competing phase. 
\end{observation}


\begin{lemma}\label{lemma:rebalancing}
  When node $v_2$ completes the rebalancing phase, $v_2$ has sent
  exactly $\id_2+2$ pulses to node $v_1$ and $\id_2+2$ pulses to
  node $v_3$. Moreover, at the end of the rebalancing phase, $v_2$
  has received exactly $\id_2+2$ pulses from $v_1$ and $\id_2+2$
  pulses from $v_3$.
\end{lemma}
\begin{proof}
  Once $v_2$ terminates its competing phase it has sent/received
  $\id_{2}$ pulses to/from each of its neighbors. At
  Line~\ref{line-algo-unoriented-check-alone-beg}, it sends $2$
  pulses to $v_1$ and at
  Line~\ref{line-algo-unoriented-rebalancing-beg} it sends $2$
  pulses to $v_3$. Moreover, at
  Lines~\ref{line-algo-unoriented-check-alone-rec1}
  and~\ref{line-algo-unoriented-rebalancing-end}, it has received two
  pulses from $v_3$ and at
  Lines~\ref{line-algo-unoriented-check-alone-end}
  and~\ref{line-algo-unoriented-rebalancing-rec0} it has received two
  pulses from $v_1$ (recall, Lemma \ref{lemma:nomaxnoleader}, that if
  Line~\ref{line-algo-unoriented-rebalancing-rec0} is executed then
  $q=0$ at Line~\ref{line-algo-unoriented-check-alone-end}).
\end{proof}

\begin{observation}\label{obs:untilrebalancing}
  From Observation \ref{obs:reducing} and Lemma \ref{lemma:rebalancing},
  $v_1$ and $v_3$ cannot distinguish during the first $\id_2+2$
  iterations of their competing phase if they are in the original ring
  or in the reduced ring.
\end{observation}

We are now ready to prove Proposition \ref{prop:redux}. 
\begin{proof}[Proof of Proposition \ref{prop:redux}]
  By Observation \ref{obs:untilrebalancing}, before node $v_2$
  completes its rebalancing phase, nodes $v_1$ and $v_3$ are not
  able to distinguish this execution from the one on the reduced
  ring. By Lemma \ref{lemma:nomaxnoleader}, node $v_2$ completes its rebalancing phase and it eventually enters the relaying
  phase. It is easy to see that, from this point onward, $v_2$ behaves as
  an asynchronous link between $v_1$ and $v_3$ as long as $v_2$ does not
  exit the loop at Lines~\ref{line-algo-unoriented-relaying-loop-beg}--\ref{line-algo-unoriented-relaying-end}.
\end{proof}

From Proposition \ref{prop:redux}, it follows that as long as no node terminates and there are at least two active nodes, the one with the minimum identifier will enter the relaying phase, and the other active nodes will behave as if this node were not present in the ring.

Observe that the global minimality of $\id_{2}$ is used only in the proof of Lemma~\ref{lemma:nomaxnoleader}, to establish that $v_2$ must enter the relaying phase, all the other lemmas, when needed, they only require the identifier of $v_2$ to be a local minimum among $v_1,v_2,v_3$. As a matter of fact, a node $v'$ may have an identifier that is a local minimum over a sufficiently long sequence of nodes, allowing it, in some executions, to enter the relaying phase before the global minimum. This does not affect the indistinguishability claimed in Proposition~\ref{prop:redux} that remains valid for $v'$ and its two neighbours.

We now show that no node terminates before the one with the maximum
identifier exits its competing phase. Moreover, we show that this
node eventually enters the global termination phase, triggering the
algorithm's termination with the election of the node with the
maximum identifier. We will show that each node terminates when it has
received three more pulses on one port than the other. In the next
lemma, we characterize when such an event can happen.

\begin{lemma}\label{lemma:sent-received-3-relay}
  During the execution, if at some point, some node $v$ has sent
  three more pulses on port $q$ than it has sent
  on port $1-q$, or if $v$ has sent three more
  pulses on port $q$ than it has received on
  port $q$, then either $v$ has executed Line~\ref{line-algo-unoriented-termination-beg},
  or $v$ is a relaying node and $v$ has previously received three more
  pulses on port $1-q$ than on port $q$.

  In any case, when this happens, the node $v$ does
  not send any pulses afterwards.
\end{lemma}

\begin{proof}
  During the competing phase, on each port, each node alternately 
  sends a pulse on each port and receives a pulse on each port.
  At the end of the competing phase, node with identifier $\id$
  has sent and received $\id$ pulses on each port.  At
  Line~\ref{line-algo-unoriented-check-alone-beg}, it sends two
  pulses on port $0$. Then, if the condition at
  Line~\ref{line-algo-unoriented-test-alone} is true, it sends a third
  pulse on port $0$ (i.e., it executes
  Line~\ref{line-algo-unoriented-termination-beg}) and stops sending
  pulses afterwards.

  Suppose now that the condition at Line~\ref{line-algo-unoriented-test-alone} is not satisfied. Then it
  means that at Line~\ref{line-algo-unoriented-check-alone-end}, $q = 0$ and at this point, $v$ has sent $\id+2$
  pulses on port $0$, has sent $\id$ pulses on port $1$, has
  received $\id+1$ pulses on port $0$, and has received $\id+1$
  pulses on port $1$. After the rebalancing phase (Lines~\ref{line-algo-unoriented-rebalancing-beg} to~\ref{line-algo-unoriented-rebalancing-end}),
  $v$ has sent and received $\id+2$ pulses on each port. 

  After that, it just relays the pulses in the order they arrive and
  stores the difference between the number of pulses sent on each port in the
  variable $\diff$. Thus, after each iteration of the loop of the
  relaying phase, the number of pulses it has sent on a port $q$
 is precisely the number of pulses it has received
  on port $1-q$.  Thus, it will send three more
  pulses on port $q$ than on port $1-q$ only if it has received three more pulses on port $1-q$
  than on port $q$. When this
  happens, $|\diff| = 3$ and node $v$ exits the loop and
  terminates.
\end{proof}

\begin{corollary}\label{cor:termination}
  If a node $v$ terminates the algorithm (i.e., executes
  Line~\ref{line-algo-unoriented-return-leader}
  or~\ref{line-algo-unoriented-return-nonleader}), then there is a
  node $v'$ that has executed
  Line~\ref{line-algo-unoriented-termination-beg}, and $v'$ will
  eventually terminate in the $\leader$ state.
\end{corollary}

\begin{proof}
  By contradiction, suppose that some node $v$ terminates and that
  no node has executed Line~\ref{line-algo-unoriented-termination-beg} before. Then $v$ necessarily
  terminates the algorithm at
  Line~\ref{line-algo-unoriented-return-nonleader}, and thus at some
  time $t$, $v$ has received three more pulses on some port $q$ than
  on port $1-q$. Among all such nodes, consider the node $v$
  such that this time $t$ is minimal. Let $v'$ be the neighbor of $v$
  behind port $q$. By Lemma~\ref{lemma:sent-received-3-relay}, $v$ is
  a relaying node, and $v$ has received three more pulses from
  $v'$ than it has sent to $v'$.  Consequently, there exists a time
  $t' <t$ where $v'$ has sent at least three more pulses to $v$ than
  it has received from $v$. By Lemma~\ref{lemma:sent-received-3-relay}
  and since we assumed that $v'$ has not executed Line~\ref{line-algo-unoriented-termination-beg}, it implies
  that at time $t'$, $v'$ has received three more pulses on one port
  than on the other. But this contradicts the definition of $v$ and
  $t$.
\end{proof}

By iteratively applying Corollary~\ref{cor:loser} and Proposition~\ref{prop:redux}, and by observing that, by Corollary~\ref{cor:termination}, no node terminates early, we obtain the following corollary:

\begin{corollary}\label{cor:loser-all}
  The node $v$ with the maximum identifier eventually becomes the only
  competing node. 
  Moreover, any other node eventually becomes a relaying node and
  cannot terminate the algorithm in the $\leader$ state.
\end{corollary}

We now show that the node with the maximum identifier eventually enters the $\leader$ state and that, after this point, no more pulses are sent in the network, and every other node eventually enters the $\nonleader$ state.

\begin{lemma}\label{lemma:termundir}

  In a ring containing only one competing node $v$, the condition
  at Line~\ref{line-algo-unoriented-test-alone} is true for
  $v$. Moreover, $v$ will eventually execute
  Lines~\ref{line-algo-unoriented-termination-beg},
  \ref{line-algo-unoriented-termination-end},
  and~\ref{line-algo-unoriented-return-leader}, terminating the
  algorithm in the $\leader$ state. Furthermore, when $v$ executes Line~\ref{line-algo-unoriented-return-leader}, there are no pulses in transit and  for every other node $v'$, we have $|\diff(v')| = 3$ and thus  $v'$ eventually terminates in the $\nonleader$ state.

\end{lemma}

\begin{proof}
  Suppose that there is a moment $t$ where the ring contains only one
  competing node $v$. Then, by Proposition~\ref{prop:redux},
  Corollary~\ref{cor:loser}, and Corollary~\ref{cor:termination}
  applied iteratively to all nodes of the ring different from $v$,
  node $v$ cannot distinguish the original ring from a ring
  containing only $v$.

  Thus during the competing phase, $v$ sends and receives $\id(v)$
  pulses to itself in each direction. In the solitude-checking
  phase, $v$ sends two pulses on port $0$ that are eventually
  delivered to itself on port $1$. Once $v$ has received both
  pulses, the condition at
  Line~\ref{line-algo-unoriented-test-alone} is true and $v$ enters
  the global termination phase. It sends another pulse to itself on
  port $0$ (Line~\ref{line-algo-unoriented-termination-beg}) that is
  eventually delivered on port $1$ and $v$ can then execute
  Lines~\ref{line-algo-unoriented-termination-end}
  and~\ref{line-algo-unoriented-return-leader} terminating in the
  $\leader$ state.

  Note that by Corollary~\ref{cor:loser}, when $v$ exits its competing
  phase, every other node $v'$ is a relaying node. Consequently,
  at this point, it has sent and received the same number of pulses
  on each port by Observation \ref{observation:obvious1}. Therefore,
  at this point, there are no pulses in transit in the network and,
  by Lemma~\ref{lemma:rebalancing}, for each node $v' \neq v$, we
  have $\diff(v') = 0$.

  Once $v$ has executed
  Line~\ref{line-algo-unoriented-check-alone-beg}, there are two
  pulses in transit in the network that are moving in the same
  direction and each node $v' \neq v$ relays them (through the same
  port) before they reach $v$ on port $1$. So when $v$ executes
  Line~\ref{line-algo-unoriented-check-alone-end}, for each node
  $v' \neq v$, we have $|\diff(v')| = 2$. Then, when $v$ executes
  Line~\ref{line-algo-unoriented-termination-beg}, it sends a third
  pulse in the same direction that will eventually reach $v$ on port
  $1$, at which point there are no pulses in transit.  In the
  meantime, every node $v' \neq v$ will have relayed it and its
  variable $\diff(v')$ will satisfy $|\diff(v')| = 3$. Thus, after
  $v'$ has relayed this third pulse, $v'$ will exit the loop in the
  relaying phase and it will terminate in the $\nonleader$ state.
\end{proof}

\begin{proof}[Proof of \Cref{thm:main2}]
It suffices to show that Algorithm~\ref{algo:unoriented} is a correct leader election algorithm for all unoriented rings and it has a message complexity of $O(n \id_{\max})$.

  The correctness of the algorithm follows directly from Corollary
  \ref{cor:loser-all} and Lemma \ref{lemma:termundir}.  Regarding the
  message complexity, take the node $v_{\max}$ with maximum
  identifier $\id_{\max}$, by Corollary \ref{cor:loser-all} and
  Lemma \ref{lemma:termundir}, this node will execute Lines \ref{line-algo-unoriented-init}--\ref{line-algo-unoriented-return-leader}. It
  is immediate that the number of pulses sent executing these lines
  is $4\id_{\max}+3$ (the factor 4 comes from the fact that we
   initially double each identifier and send pulses in both directions).  Consider now any other node
  $v_j$ with identifier $\id_{j}\neq \id_{\max}$, by Corollary
  \ref{cor:loser} and Lemma \ref{lemma:rebalancing}, $v_j$ first
  executes Lines \ref{line-algo-unoriented-init}--\ref{line-algo-unoriented-check-alone-end} and Lines \ref{line-algo-unoriented-rebalancing-beg}--\ref{line-algo-unoriented-rebalancing-end}, the number of pulses sent at
  these lines is $4\id_{j}+4$, and then it executes Lines \ref{line-algo-unoriented-relaying-loop-beg}--\ref{line-algo-unoriented-relaying-end}, but
  here it will only send a pulse when it receives one. Therefore,
  Lemma \ref{lemma:termundir} and the fact that $v_{\max}$ sends
  $4\id_{\max}+3$ pulses show that $v_j$ cannot send more than
  $\delta_{j}=(4\id_{\max}+3)-(4\id_{j}+4)$ pulses while
  relaying. This implies that the total number of pulses, and thus messages, sent by
  Algorithm \ref{algo:unoriented} is at most $n(4\id_{\max}+3)$.
\end{proof}

\section{Conclusions and Open Questions}
In this paper, we demonstrate that non-uniform  quiescently terminating content-oblivious leader election is achievable in general 2-edge-connected networks, and that uniform quiescently
 terminating content-oblivious leader election is achievable in unoriented rings. 
 
Previously, such leader election was only known to be possible in ring topologies---specifically, a quiescently terminating algorithm in oriented rings and a stabilizing algorithm in unoriented rings~\cite{content-oblivious-leader-election-24}. Consequently, we remove the need for a preselected leader in the general algorithm simulation result by Censor-Hillel, Cohen, Gelles, and Sela~\cite{fully-defective-22}, trading this requirement for the much weaker assumption of non-uniformity. For the special case of ring topologies, we provide a definitive answer on the computational equivalence of fully-defective communication and noiseless communication.

Several intriguing open questions remain. The foremost is whether 
\emph{uniform} content-oblivious leader election algorithm---one that requires no prior knowledge of the network---is possible for general 2-edge-connected toplogies.
So far, this has only been established for ring topologies—first for oriented rings~\cite{content-oblivious-leader-election-24} and later for unoriented rings (\Cref{thm:main2}). For general 2-edge-connected topologies, our current algorithm requires an upper bound $N$ on the number of nodes $n$, as node identifiers must be spaced at least $N$ apart to ensure that no node other than the leader $r$ satisfies the $\leader$ exit condition $\counter(v) = \id(v)$. This reliance on $N$ seems inherent to our counting-based approach, in which counters of adjacent nodes may differ by only one.
\color{black}

Another important open question is whether the efficiency of our algorithms can be improved. 
Our algorithm for general 2-edge-connected topologies is highly sequential due to its use of DFS.
 A common way to measure \emph{time complexity} of an asynchronous algorithm is to assume that each message takes one unit of time to be delivered~\cite{attiya2004distributed}. Under this model, the time complexity of our algorithm is $O(m \cdot N \cdot \idmin)$, matching its message complexity, as the DFS traverses the $m$ edges sequentially, and the traversal of each edge requires sending $O(N \cdot \idmin)$ pulses. It remains an open question whether we can reduce this time complexity by making the algorithm less sequential.

Much less is known about message complexity lower bounds for content-oblivious leader election. The only known result is a lower bound of $\Omega\left(n \cdot \log \frac{k}{n}\right)$, where $k$ is the number of distinct, assignable identifiers in the network, due to Frei, Gelles, Ghazy, and Nolin~\cite{content-oblivious-leader-election-24}. Currently, there remains a significant gap between this lower bound and the upper bounds, both in rings and in general 2-edge-connected networks. A key open question is whether the message complexity for general 2-edge-connected networks is inherently higher than that of ring topologies. If so, how can we leverage the structural properties of the hard instances to establish such a lower bound?

\printbibliography
\end{document}